\documentclass{article}
\usepackage[a4paper, top=1.25in, bottom=1.25in, left=1.25in, right=1.25in]{geometry}
\usepackage[utf8]{inputenc}
\usepackage[usenames,dvipsnames]{xcolor}

\usepackage{array}
\usepackage[font=footnotesize,labelfont=bf]{caption}
\usepackage{textcomp}
\usepackage{stfloats}
\usepackage{url}
\usepackage{verbatim}
\usepackage{graphicx}
\usepackage[noadjust]{cite}
\usepackage[normalem]{ulem}

\usepackage{amsfonts,amssymb,amsthm}
\usepackage{enumitem}
\usepackage{mathtools}
\usepackage{braket}
\usepackage{algorithm,algpseudocode}
\usepackage{siunitx}
\usepackage{bm}
\usepackage{mleftright}
\usepackage{booktabs}

\usepackage{color}
\usepackage{hyperref}
\hypersetup{colorlinks=true,
    linkcolor=BrickRed,
    citecolor=OliveGreen,
    urlcolor=MidnightBlue}
    
\usepackage{cleveref}

\crefname{figure}{}{}
\crefrangeformat{thm}{#3#1#4--#5#2#6}
\crefrangeformat{exmp}{#3#1#4--#5#2#6}

\usepackage{pgfplots}
\pgfplotsset{compat=1.9}
\usepgfplotslibrary{fillbetween}

\DeclareMathOperator*{\minimize}{min}
\DeclareMathOperator*{\maximize}{max}
\DeclareMathOperator*{\subjto}{subj. to}

\DeclareMathOperator*{\argmin}{arg\,min}

\DeclareMathOperator{\proj}{proj}
\DeclareMathOperator{\tr}{tr}
\DeclareMathOperator{\closure}{cl}
\DeclareMathOperator{\interior}{int}
\DeclareMathOperator{\relinterior}{relint}

\DeclareMathOperator{\domain}{dom}
\DeclareMathOperator{\diag}{diag}

\DeclareMathOperator{\supp}{supp}

\DeclarePairedDelimiterX{\divx}[2]{(}{)}{#1\mspace{1.5mu}\delimsize\|\mspace{1.5mu}#2}
\DeclarePairedDelimiterX{\divy}[2]{(}{)}{#1\mspace{1mu}\delimsize|\mspace{1mu}#2}
\DeclarePairedDelimiterX{\inp}[2]{\langle}{\rangle}{#1, #2}
\DeclarePairedDelimiterX{\norm}[1]{\lVert}{\rVert}{#1}
\DeclarePairedDelimiterX{\abs}[1]{\lvert}{\rvert}{#1}
\DeclarePairedDelimiterX{\bk}[2]{\langle}{\rangle}{#1 \delimsize\vert #2}

\makeatletter
\def\ALG@special@indent{%
    \ifdim\ALG@thistlm=0pt\relax
        \hskip-\leftmargin
    \else
        \hskip\ALG@thistlm
    \fi
}
\newcommand{\Input}[1]{\item[]\noindent\ALG@special@indent \textbf{Input:}\ #1}
\newcommand{\Output}[1]{\item[]\noindent\ALG@special@indent \textbf{Output:}\ #1}
\newcommand{\Indent}[1]{\item[]\noindent\ALG@special@indent \hspace{2.675em} #1}
\makeatother

\newcommand*\wc{{\mkern 2mu\cdot\mkern 2mu}}
\newcommand{\newinf}{\mathop{\mathrm{inf}\vphantom{\mathrm{sup}}}}

\NewDocumentCommand{\grad}{e{_^}}{%
  \mathop{}\!% \mathop for good spacing before \nabla
  \nabla
  \IfValueT{#1}{_{\!#1}}% tuck in the subscript
  \IfValueT{#2}{^{#2}}% possible superscript
}

\newtheorem{thm}{Theorem}[section]
\newtheorem{lem}[thm]{Lemma}
\newtheorem{prop}[thm]{Proposition}
\newtheorem{cor}[thm]{Corollary}
\newtheorem{defn}[thm]{Definition}
\newtheorem{rem}[thm]{Remark}
\newtheorem{exmp}[thm]{Example}

\newcommand{\footremember}[2]{%
    \footnote{#2}
    \newcounter{#1}
    \setcounter{#1}{\value{footnote}}%
}
\newcommand{\footrecall}[1]{%
    \footnotemark[\value{#1}]%
} 

\begin{document}

\title{A Bregman Proximal Perspective on Classical and Quantum Blahut-Arimoto Algorithms}

\author{%
    Kerry He\footremember{monash}{Department of Electrical and Computer System Engineering, Monash University, Clayton VIC 3800, Australia. \url{{kerry.he1, james.saunderson}@monash.edu}} \and James Saunderson\footrecall{monash} \and Hamza Fawzi\footremember{cambridge}{Department of Applied Mathematics and Theoretical Physics, University of Cambridge, Cambridge CB3 0WA, United Kingdom. \url{h.fawzi@damtp.cam.ac.uk}}
}
\date{}

\maketitle

\begin{abstract}
    The Blahut-Arimoto algorithm is a well-known method to compute classical channel capacities and rate-distortion functions. Recent works have extended this algorithm to compute various quantum analogs of these quantities. In this paper, we show how these Blahut-Arimoto algorithms are special instances of mirror descent, which is a type of Bregman proximal method, and a well-studied generalization of gradient descent for constrained convex optimization. Using recently developed convex analysis tools, we show how analysis based on relative smoothness and strong convexity recovers known sublinear and linear convergence rates for Blahut-Arimoto algorithms. This Bregman proximal viewpoint allows us to derive related algorithms with similar convergence guarantees to solve problems in information theory for which Blahut-Arimoto-type algorithms are not directly applicable. We apply this framework to compute energy-constrained classical and quantum channel capacities, classical and quantum rate-distortion functions, and approximations of the relative entropy of entanglement, all with provable convergence guarantees. 
\end{abstract}

\section{Introduction}

Channel capacities, which quantify the maximum rate at which information can be transmitted through a noisy channel, are typically defined using variational expressions that do not admit closed-form solutions. Many important capacities also have the property of being defined as the solutions of convex optimization problems, and therefore convex optimization techniques can be used to efficiently compute these quantities. One of the most well-known algorithms for estimating classical channel capacities is the Blahut-Arimoto algorithm~\cite{arimoto1972algorithm,blahut1972computation}. More recently, several works have extended this method to compute various quantum channel capacities~\cite{nagaoka1998algorithms,osawa2001numerical,li2019blahut,li2019computing, ramakrishnan2020computing}. The classical Blahut-Arimoto algorithm can be interpreted in various ways in terms of existing algorithmic schemes for convex optimization, most generally as a mirror descent algorithm, which is a type of Bregman proximal method (see Section~\ref{sec:related-work}). The goal of this work is to provide a systematic interpretation of the quantum Blahut-Arimoto algorithm and its convergence analysis using tools from convex optimization. Such an interpretation is not only useful from a theoretical point of view, but it also allows us to draw on the literature from convex optimization to effectively deal with more complicated optimization setups encountered in quantum information theory.

First, we relate the convergence analysis of Blahut-Arimoto algorithm to the general notions of relative smoothness and relative strong convexity that have been proposed in~\cite{bauschke2017descent,lu2018relatively,teboulle2018simplified}. While traditional convergence analyses for first-order methods are based on the assumption that the objective function has a Lipschitz-continuous gradient, with stronger bounds on the convergence rates being available if the function is additionally strongly convex, these assumptions are not satisfied for many applications in information theory. The notions of relative smoothness and relative strong convexity have recently been proposed as sufficient conditions to guarantee sublinear and linear convergence rates of Bregman proximal methods. We establish in this paper a direct link between these notions and the convergence analysis of Blahut-Arimoto algorithms.

Second, a limitation with current Blahut-Arimoto algorithms is that they are unable to effectively deal with general linear constraints, which arise in the computation of energy-constrained channel capacities or rate-distortion functions. Existing methods typically dualize the constraints by minimizing the Lagrangian and treating the Lagrange multipliers as constants~\cite{blahut1972computation}, however it is not clear how to choose the Lagrange multipliers to solve for a particular set of constraints. Other works require repeatedly solving a sequence of inner or outer subproblems~\cite{li2019blahut,hayashi2022bregman}, which can be computationally inefficient. Instead, a natural extension of mirror descent to account for general linear constraints is to use a primal-dual hybrid gradient (PDHG) method~\cite{pock2009algorithm,chambolle2011first,chambolle2016ergodic,jiang2022bregman}. This algorithm is a variation of mirror descent which solves saddle-point problems by using alternating mirror descent and ascent steps, and can be applied to linearly constrained convex optimization problems by simultaneously solving for the optimal primal and dual variables of a Lagrangian function.

\subsection{Related Work} \label{sec:related-work}

% There are these classical BAA interpretations
% As the quantum BAA are new, the same depth of interpretations does not exist for theese algorithms yet

As a well-established method, many interpretations of the classical Blahut-Arimoto algorithm have been developed. The most well-known interpretation of the algorithm is as an alternating optimization algorithm~\cite{cover1999elements,o1998alternating}. Notably, in~\cite{o1998alternating}, it was shown how the alternating optimization framework unifies several important algorithms in information theory, including the expectation-maximization algorithm (see also~\cite{hayashi2023minimization}). In~\cite{vontobel2008generalization}, the Blahut-Arimoto algorithm is given an interpretation as a method which introduces a ``surrogate'' function which approximates the original objective function, and which is easier to optimize. In particular, the authors show how a suitable surrogate function for the classical channel capacity problem can be found by decomposing the objective function into the difference of two concave functions, then linearizing one part of this decomposition. In~\cite{matz2004information,naja2009geometrical}, it was shown how the Blahut-Arimoto algorithm for classical channel capacities can be interpreted as mirror descent, which is a well-studied type of first-order Bregman proximal method~\cite{nemirovskij1983problem,beck2003mirror,tseng2008accelerated,beck2017first}, and often regarded as a generalization of projected gradient descent. It turns out that the alternating optimization and surrogate function interpretations of the Blahut-Arimoto algorithm are also closely related to Bregman proximal methods. In particular, expectation-maximization was given a mirror descent interpretation in~\cite{kunstner2021homeomorphic}, and the difference-of-convex algorithm was interpreted through the lens of Bregman proximal methods in~\cite{faust2023bregman}. The quantum Blahut-Arimoto algorithm, being more recently developed method, currently lacks the same range of interpretations as its classical counterpart.

Apart from Blahut-Arimoto-style approaches, other algorithms proposed to solve for quantum channel capacities include~\cite{shor2003capacities,hayashi2005qubit,sutter2015efficient}. Alternatively, it was shown in~\cite{chandrasekaran2017relative,fawzi2019semidefinite,fawzi2018efficient} how certain quantum channel capacities could be formulated as instances of the more general class of quantum relative entropy programs, however there are currently no practical algorithms able to solve large-scale problem instances of this class. Several works have applied mirror descent to solve various problems in quantum information theory, including quantum state tomography~\cite{li2019convergence} and minimization of quantum R\'enyi divergences~\cite{you2022minimizing}. However, neither of these works establish convergence rates. Other works which have applied bespoke first-order methods to solve particular problems in quantum information theory include~\cite{zinchenko2010numerical,drusvyatskiy2015projection,knee2018quantum,youssry2019efficient,winick2018reliable,hayashi2022bregman}.

\subsection{Main Results}

In this paper, we extend the mirror descent interpretation of classical Blahut-Arimoto algorithms by~\cite{matz2004information,naja2009geometrical} in several ways. First, we show how the Blahut-Arimoto algorithm in~\cite{ramakrishnan2020computing}, which generalizes classical~\cite{blahut1972computation,arimoto1972algorithm} and quantum~\cite{nagaoka1998algorithms,osawa2001numerical,li2019blahut,li2019computing} Blahut-Arimoto channel capacity algorithms, produces iterates which are identical to those generated by an entropic version of mirror descent when both methods are applied to solving the same optimization problem. Second, we show how certain properties of information measures, which are used to prove convergence of these Blahut-Arimoto algorithms, can instead be interpreted as relative smoothness and strong convexity properties. Using this, we can recover known convergence rates of Blahut-Arimoto algorithms using standard convex analysis techniques. These results are formally presented in Theorem~\ref{thm:ba-md}.

Next, we show how a primal-dual extension of mirror descent, PDHG, allows us to solve problems with general linear constraints. Additionally, we prove ergodic sublinear convergence of PDHG by using relative smoothness analysis. By using PDHG with different kernel functions, we derive new methods to solve problems which Blahut-Arimoto algorithms have previously not been used to solve, or were unable to efficiently solve. Specifically, we show how entropic PDHG can be used to compute energy-constrained classical and quantum channel capacities, as well as classical and quantum rate-distortion functions. We also show that although entropic PDHG is not a suitable method to compute the relative entropy of a quantum resource, PDHG using a negative log-determinant kernel function to compute this quantity results in an algorithm with provable convergence guarantees.

\section{Preliminaries} \label{sec:prelim}
\subsection{Notation}\label{secsub:notation}
Let $\bar{\mathbb{R}}\coloneqq\mathbb{R}\cup\{+\infty\}$ be the set of extended real numbers, let $\mathbb{N}$ be the set of positive integers (excluding zero), let $\mathbb{V}$ represent a finite-dimensional inner product space, let $\mathbb{H}^n$ represent the set of $n\times n$ Hermitian matrices with trace inner product $\inp{X}{Y} = \tr[X^\dag Y]$ where $X^\dag$ is the adjoint of $X$, and let $\preceq$ denote the semidefinite order for Hermitian matrices. Also, let $\mathbb{R}^n_+$ denote the non-negative orthant, $\mathbb{H}^n_+$ the set of positive semidefinite Hermitian matrices, and $\mathbb{R}^n_{++}$ and $\mathbb{H}^n_{++}$ denote the respective interiors. We denote the $n$-dimensional probability simplex as 
\begin{equation*}
    \Delta_n \coloneqq \biggl\{p \in \mathbb{R}^n_+ : \sum_{i=1}^n p_i = 1 \biggr\},
\end{equation*}
and set of $n\times m$ column stochastic matrices as 
\begin{equation*}
    \mathcal{Q}_{n,m} \coloneqq \biggl\{ Q\in\mathbb{R}^{n\times m}_+ : \sum_{i=1}^n Q_{ij} = 1, \forall j=1,\ldots,m \biggr\}.
\end{equation*}
We use $\interior$, $\relinterior$, $\closure$, and $\domain$ to denote the interior of a set, the relative interior (see e.g.~\cite[Section 2.1.3]{boyd2004convex}) of a set, the closure of a set, and the domain of a function, respectively. 

On an $n$-dimensional Hilbert space $\mathcal{H} \cong \mathbb{C}^n$, we denote the set of bounded self-adjoint linear operators as $\mathcal{B}(\mathcal{H}) \cong \mathbb{H}^n$, the set of bounded nonnegative and positive self-adjoint linear operators as $\mathcal{B}(\mathcal{H})_+ \cong \mathbb{H}^n_+$ and $\mathcal{B}(\mathcal{H})_{++} \cong \mathbb{H}^n_{++}$, respectively, and the set of density matrices as 
\begin{equation*}
    \mathcal{D}(\mathcal{H}) \coloneqq \{\rho \in \mathcal{B}(\mathcal{H})_+ : \tr[\rho] = 1\}.
\end{equation*}
We will also define $\mathcal{D}(\mathcal{H})_{++}\coloneqq\relinterior\mathcal{D}(\mathcal{H})=\{\rho \in \mathcal{B}(\mathcal{H})_{++} : \tr[\rho] = 1\}$. Where required, we will use subscripts to clarify that a quantum state exists in a particular system. For instance, we write $\rho_A$ to indicate the density matrix is in system $\mathcal{H}_A$. For a bipartite system $\rho_{AB}\in\mathcal{D}(\mathcal{H}_A\otimes\mathcal{H}_B)$, the partial trace over $\mathcal{H}_A$ (see, e.g.,~\cite[Definition 4.3.4]{wilde2017quantum}) is denoted $\tr_A(\rho_{AB})$.

The set of completely positive trace preserving (CPTP) quantum channels (see, e.g.~\cite[Definitions 4.4.2 and 4.4.3]{wilde2017quantum}) with input system $\mathcal{H}_A$ and output system $\mathcal{H}_B$ is denoted as $\Phi(\mathcal{H}_A, \mathcal{H}_B)$. For a quantum channel $\mathcal{N}\in\Phi(\mathcal{H}_A, \mathcal{H}_B)$, the Stinespring representation is $\mathcal{N}(\rho)=\tr_E(U\rho U^\dag)$ for some isometry $U: \mathcal{H}_A \rightarrow \mathcal{H}_B \otimes \mathcal{H}_E$, where $\mathcal{H}_E$ denotes an auxiliary environment system. The corresponding complementary channel $\mathcal{N}_{\mathrm{comp}}\in\Phi(\mathcal{H}_A, \mathcal{H}_E)$ is defined as $\mathcal{N}_{\mathrm{comp}}(\rho) = \tr_B(U\rho U^\dag)$. For any linear operator $\mathcal{A}:\mathbb{V}\rightarrow\mathbb{V}'$ between inner product spaces, the corresponding adjoint operator $\mathcal{A}^\dag:\mathbb{V}'\rightarrow\mathbb{V}$ is defined to satisfy $\inp{y}{\mathcal{A}(x)} = \inp{\mathcal{A}^\dag(y)}{x}$ for all $x\in\mathbb{V}$ and $y\in\mathbb{V}'$. The norm induced by the inner product of a vector space $\mathbb{V}$ is denoted as $\norm{\wc}_2 \coloneqq \sqrt{\inp{\wc}{\wc}}$ (i.e., Euclidean norm for $\mathbb{R}^n$, Frobenius norm for $\mathbb{H}^n$). The identity matrix is denoted as $\mathbb{I}$, and the all ones vector is denoted $\bm{1}$. The largest and smallest eigenvalues of $A\in\mathbb{H}^n$ are denoted as $\lambda_{\mathrm{max}}(A)$ and $\lambda_{\mathrm{min}}(A)$, respectively. We denote the $j$-th column of a matrix $A\in\mathbb{R}^{n\times m}$ as $A_j$.

\subsection{Classical and Quantum Entropies}\label{subsec:entropy}

The Shannon entropy for a random variable $X$ with probability distribution $x\in\Delta_n$ is 
\begin{equation}
    H(x)\coloneqq-\sum_{i=1}^n x_i \log(x_i),
\end{equation}
and the Kullback-Leibler (KL) divergence between random variables $X$ and $Y$ with probability distributions $x,y\in\Delta_n$ is 
\begin{equation}
    H\divx{x}{y}\coloneqq\sum_{i=1}^n x_i \log(x_i / y_i).
\end{equation}
For a joint distribution $P\in\Delta_{n\times m}$ on random variables $X$ and $Y$ and marginal distribution $p\in\Delta_m$ on $Y$ where $p_j=\sum_{i=1}^nP_{ij}$ for all $j$, the classical conditional entropy is 
\begin{equation}
    H\divy{X}{Y}_P \coloneqq H(P) - H(p).
\end{equation}
Similarly, the von Neumann entropy for density matrix $\rho\in\mathcal{D}(\mathcal{H})$ is
\begin{equation}
    S(\rho)\coloneqq-\tr[\rho \log(\rho)],
\end{equation}
and the quantum relative entropy between quantum states $\rho,\sigma\in\mathcal{D}(\mathcal{H})$ is
\begin{equation}
    S\divx{\rho}{\sigma}\coloneqq\tr[\rho (\log(\rho) - \log(\sigma))].
\end{equation}
For a bipartite density matrix $\rho\in\mathcal{D}(\mathcal{H}_A\otimes\mathcal{H}_B)$, the quantum conditional entropy is
\begin{equation}
    S\divy{A}{B}_{\rho}\coloneqq S(\rho) - S(\tr_A(\rho)).
\end{equation}
These entropies and divergences are defined using the convention that $0\log(0)=0$ and $-x\log(0)=+\infty$ for all $x>0$. For pairs of quantum states, this implies that $S\divx{\rho}{\sigma}=+\infty$ when $\supp \rho \nsubseteq \supp \sigma$~\cite{watrous2018theory}, where $\supp$ denotes the support of a matrix.

It is well known that Shannon and von Neumann entropy are concave functions. The KL divergence and quantum relative entropy are jointly convex on their two respective arguments~\cite[Corollary 11.9.2]{wilde2017quantum}. Additionally, classical and quantum conditional entropies are concave with respect to $P$ and $\rho$, respectively~\cite[Corollary 7.27]{holevo2019quantum}.

\section{Mirror Descent and Relative Smoothness} \label{sec:mirror-descent}

For a convex function $f:\mathbb{V}\rightarrow\bar{\mathbb{R}}$ and convex set $\mathcal{C}\subseteq\closure\domain f$, we consider the constrained convex optimization problem
\begin{equation}
    \minimize_{x \in \mathcal{C}} \, f(x). \label{eqn:constr-min}
\end{equation}
Throughout this paper, unless otherwise stated, we assume that the domain of $f$ has non-empty interior, and that $f$ is differentiable on the interior of its domain. Many problems in quantum information theory can be cast in the form~\eqref{eqn:constr-min}, where the decision variable is a density matrix. As such, the dimensions of these problems typically scale exponentially with the number of qubits, making it desirable to solve these problems using first-order methods. However, gradient descent methods inherently impose a Euclidean structure on a problem by measuring distances between points using the Euclidean norm, whereas probability distributions and density matrices are more naturally compared using quantities such as relative entropies. Mirror descent resolves these issues by replacing the Euclidean norm that appears in gradient descent with a Bregman divergence, which adapts it to geometries more suitable for given objectives $f$ and sets $\mathcal{C}$. The Bregman divergence is defined using a special class of kernel functions, which we introduce below.
\begin{defn}[{Legendre function~\cite[Definition 2.1]{teboulle2018simplified}}]
    A function $\varphi: \mathbb{V}\rightarrow\bar{\mathbb{R}}$ is Legendre if it is proper, lower semicontinuous, strictly convex, and essentially smooth.
\end{defn}
Note that essential smoothness of Legendre functions implies that their domains have non-empty interior, and that they are differentiable on the interior of their domains. See, e.g.,~\cite{rockafellar1970convex} for precise definitions of these terms. Using these Legendre functions, we define the Bregman divergence as follows.
\begin{defn}[{Bregman divergence~\cite{bregman1967relaxation}}]
    Consider a Legendre function $\varphi: \mathbb{V}\rightarrow\bar{\mathbb{R}}$. The Bregman divergence $D_\varphi\divx{\wc}{\wc} : \domain\varphi\times\interior\domain\varphi\rightarrow\mathbb{R}$ associated with the kernel function $\varphi$ is 
    \begin{equation}\label{eqn:breg-div}
        D_\varphi\divx{x}{y} \coloneqq \varphi(x) - \mleft(\varphi(y) + \inp{\nabla\varphi(y)}{x - y} \mright).  
    \end{equation}
\end{defn}
The Bregman divergence is not a metric as it is not necessarily symmetric in its two arguments, nor does it satisfy the triangle inequality. However, it satisfies certain desirable properties. From convexity of $\varphi$, it follows that Bregman divergences are non-negative. From strict convexity of $\varphi$, it follows that $D_\varphi\divx{x}{y}=0$ if and only if $x=y$. We introduce a few important Bregman divergences below.
\begin{exmp}\label{exmp:energy}
    The Bregman divergence associated with the energy function $\varphi(x)=\norm{x}_2^2/2$ with domain $\domain\varphi=\mathbb{R}^n$ is the squared Euclidean norm $D_\varphi\divx{x}{y} = \norm{x - y}_2^2/2$.
\end{exmp}
\begin{exmp}\label{exmp:cre}
    The Bregman divergence associated with negative Shannon entropy $\varphi(x)=-H(x)$ with domain $\domain\varphi=\mathbb{R}^n_+$ is the normalized KL divergence $D_\varphi\divx{x}{y}= H\divx{x}{y} - \sum_{i=1}^n \mleft(x_i - y_i\mright)$.
\end{exmp}
\begin{exmp}\label{exmp:qre}
    The Bregman divergence associated with negative von Neumann entropy $\varphi(\rho)=-S(\rho)$ with domain $\domain\varphi=\mathbb{H}^n_+$ is the normalized quantum relative entropy $D_\varphi\divx{\rho}{\sigma}= S\divx{\rho}{\sigma} - \tr[\rho - \sigma]$.
\end{exmp}
\begin{exmp}\label{exmp:cld}
    The Bregman divergence associated with the negative Burg entropy $\varphi(x)=-\sum_{i=1}^n \log(x_i)$ with domain $\domain\varphi=\mathbb{R}^n_{++}$ is the Itakura-Saito distance $D_\varphi\divx{x}{y} = \sum_{i=1}^n (x_i/y_i - \log(x_i/y_i) - 1)$.
\end{exmp}
\begin{exmp}\label{exmp:qld}
    The Bregman divergence associated with the negative log determinant function $\varphi(\rho)=-\log(\det(\rho))$ with domain $\domain\varphi=\mathbb{H}^n_{++}$ is the log determinant divergence $D_\varphi\divx{\rho}{\sigma} = \tr[\rho\sigma^{-1}] - \log(\det(\rho\sigma^{-1})) - n$.
\end{exmp}
% Sometimes, an additional requirement on $\varphi$ is that it is strongly convex with respect to some norm. We note that all of the above kernel functions are strongly convex on bounded sets.

\begin{algorithm}
    \caption{Mirror descent}
    \begin{algorithmic}
        \Input{Objective function $f$, Legendre kernel function $\varphi$, primal domain $\mathcal{C}\subseteq\domain\varphi$, step size $t_k>0$}
        \State \textbf{Initialize:} Initial point $x^0\in\relinterior\mathcal{C}$
        \For{$k=0,1,\ldots$}{}
            \State \hfill\llap{%
                    \makebox[\linewidth]{\hfill 
                    $\displaystyle x^{k+1} = \argmin_{x \in \mathcal{C}} \mleft\{ \inp{\nabla f(x^k)}{x} + \frac{1}{t_k} D_\varphi\divx{x}{x^k} \mright\}$
                    \hfill\refstepcounter{equation}\llap{(\theequation)}\label{eqn:mirror-descent}}
                } 
        \EndFor
        \Output{Approximate solution $x^{k+1}$}
    \end{algorithmic}
    \label{alg:mirror-descent}
\end{algorithm}

Using the Bregman divergence, the mirror descent algorithm for solving~\eqref{eqn:constr-min} is outlined in Algorithm~\ref{alg:mirror-descent}. The step sizes $t_k$ can be chosen to be different at every iteration. We present a good choice of step size which guarantees convergence to the optimum later in Proposition~\ref{prop:conv-rate}. The mirror descent algorithm is practical when the mirror descent update~\eqref{eqn:mirror-descent} can be efficiently computed. This is the case for certain choices of $\varphi$ and $\mathcal{C}$. When $\varphi(x)=\norm{x}_2^2/2$, this recovers projected gradient descent
\begin{align}
    x^{k+1} = \proj_\mathcal{C}(x^k - t_k \nabla f(x^k))
\end{align}
where $\proj_\mathcal{C}(y) =\argmin_{x \in \mathcal{C}} \norm{x - y}_2^2$, which can be efficiently computed for certain sets $\mathcal{C}$. When $\varphi(x)=-H(x)$ and $\mathcal{C}=\Delta_n$, then 
\begin{equation}
    x^{k+1}_i = \frac{x^k_i \exp(-t_k \partial_i f(x))}{\sum_{j=1}^n x^k_j \exp(-t_k \partial_j f(x))} \quad \textrm{ for } i=1,\ldots,n,  
    \label{eqn:exp-grad-desc}
\end{equation}
where $\partial_i f$ is the partial derivative of $f$ with respect to the $i$-th coordinate. Similarly, when $\varphi(\rho)=-S(\rho)$ and $\mathcal{C}=\mathcal{D}(\mathcal{H})$, then
\begin{equation}
    \rho^{k+1} = \frac{\exp(\log(\rho^k) - t_k\nabla f(\rho^k))}{\tr[\exp(\log(\rho^k) - t_k\nabla f(\rho^k))]}.
    \label{eqn:mtx-exp-grad}
\end{equation}
These previous two iterates~\eqref{eqn:exp-grad-desc} and~\eqref{eqn:mtx-exp-grad} are sometimes known as exponentiated or entropic mirror descent in literature. 

Essential smoothness of $\varphi$ allows it to act as a kind of barrier to its own domain, and implies that~\eqref{eqn:mirror-descent} is minimized on the interior of the domain of $\varphi$~\cite[Lemma 2.2]{maddison2021dual}. Therefore, certain constraints such as non-negativity can be elegantly encoded into the mirror descent algorithm by using suitable kernel functions such as negative entropies. This also allows mirror descent to avoid the ``stalling'' issues identified in~\cite{winick2018reliable,knee2018quantum,you2022minimizing} where gradients are not well-defined on the boundary of the domain.

\subsection{Relative Smoothness and Strong Convexity}

Convergence analysis for first-order algorithms has typically been based on the assumption that the objective function $f$ is smooth on $\mathcal{C}$ with respect to a norm $\norm{\wc}$
\begin{equation}
    \norm{\nabla f(x) - \nabla f(y)}_* \leq L \norm{x - y}, \quad \forall x,y\in\relinterior\mathcal{C},
\end{equation}
for some $L>0$, where $\norm{\wc}_*$ is the dual norm, and is strongly convex on $\mathcal{C}$
\begin{equation}
    f(x) \geq f(y) + \inp{\nabla f(y)}{x - y} + \frac{\mu}{2} \norm{x - y}^2, \quad \forall x,y\in\relinterior\mathcal{C},
\end{equation}
for some $\mu>0$. When the objective is smooth, first-order methods are typically able to converge monotonically at a sublinear rate, i.e., $O(1/k)$, where $k$ is the iteration index. When the objective is additionally strongly convex, convergence is improved to a linear rate, i.e., $O(\alpha^k)$ for some $\alpha\in(0, 1)$, where $k$ is the iteration index. However, a drawback with convergence analysis using these assumptions is many objective functions in information theory do not satisfy these properties.

Recently, these convexity properties have been extended to relative smoothness and strong convexity~\cite{bauschke2017descent, lu2018relatively, teboulle2018simplified}, which measure convexity relative to any Legendre function. 
\begin{defn}[Relative smoothness]
    A function $f$ is $L$-smooth relative to a Legendre function $\varphi$ on a convex set $\mathcal{C}$ if there exists an $L>0$ such that $L\varphi - f$ is convex on $\relinterior \mathcal{C}$.
\end{defn}
\begin{defn}[Relative strong convexity]
    A function $f$ is $\mu$-strongly convex relative to a Legendre function $\varphi$ on a convex set $\mathcal{C}$ if there exists a $\mu>0$ such that $f - \mu\varphi$ is convex on $\relinterior \mathcal{C}$.
\end{defn}
When $\varphi(x)=\norm{x}^2_2/2$, we recover the standard definitions for smoothness and strong convexity with respect to the Euclidean norm. 

There are many equivalent ways to express relative smoothness and strong convexity, as summarized in~\cite{bauschke2017descent, lu2018relatively}. Specifically, the following equivalent expressions will be useful, which all arise from equivalent conditions for convexity.
\begin{prop}\label{prop:relative-alt}
    Let $f:\mathbb{V}\rightarrow\bar{\mathbb{R}}$ be a convex function such that $f$ is differentiable on $\interior\domain f$. Let $\varphi:\mathbb{V}\rightarrow\bar{\mathbb{R}}$ be a Legendre function. Let $\mathcal{C}\subseteq\closure\domain f \cap \closure\domain\varphi$ be a convex set. The following conditions are equivalent:
    \begin{enumerate}[label=(a-\roman*),leftmargin=30pt]
        \item $L\varphi - f$ is convex on $\mathcal{C}$,
        \item $\inp{\nabla f(x) - \nabla f(y)}{x - y} \leq L (D_\varphi\divx{x}{y} + D_\varphi\divx{y}{x})$ for all $x, y \in \relinterior \mathcal{C}$,
        \item $f(x) \leq f(y) + \inp{\nabla f(y)}{x - y} + LD_\varphi\divx{x}{y}$ for all $x, y \in \relinterior \mathcal{C}$.
    \end{enumerate}
    Similarly, the following conditions are also equivalent:
    \begin{enumerate}[label=(b-\roman*),leftmargin=30pt]
        \item $f - \mu\varphi$ is convex on $\mathcal{C}$,
        \item $\inp{\nabla f(x) - \nabla f(y)}{x - y} \geq \mu (D_\varphi\divx{x}{y} + D_\varphi\divx{y}{x})$ for all $x, y \in \relinterior \mathcal{C}$,
        \item $f(x) \geq f(y) + \inp{\nabla f(y)}{x - y} + \mu D_\varphi\divx{x}{y}$ for all $x, y \in \relinterior \mathcal{C}$.
    \end{enumerate}
\end{prop}

These properties are useful as they can be used to establish convergence rates for Bregman proximal methods. In particular, we can obtain the following global convergence rates for mirror descent.
\begin{prop}[{\cite[Theorem 4.1 and Proposition 4.1]{teboulle2018simplified}}]\label{prop:conv-rate}
    Consider Algorithm~\ref{alg:mirror-descent} to solve the convex optimization problem~\eqref{eqn:constr-min}. Let $f^*$ represent the optimal value of this problem, and $x^*$ be any corresponding optimal point. If $f$ is $L$-smooth relative to $\varphi$ and $t_k=1/L$ for all $k$, then the sequence $\{f(x^k)\}$ decreases monotonically, and satisfies
    \begin{equation}
        f(x^k) - f^* \leq \frac{L}{k}D_\varphi\divx{x^*}{x^0}, \quad \forall k\in\mathbb{N}.
        \label{eqn:sublin-conv}
    \end{equation}
    If $f$ is additionally $\mu$-strongly convex relative to $\varphi$, then the sequence satisfies
    \begin{equation}
        f(x^k) - f^* \leq \mleft(1 - \frac{\mu}{L}\mright)^k LD_\varphi\divx{x^*}{x^0}, \quad \forall k\in\mathbb{N}.
        \label{eqn:lin-conv}
    \end{equation}
\end{prop}
These results show that just like gradient descent, mirror descent achieves sublinear convergence when the objective is relatively smooth, and linear convergence when the objective is additionally relatively strongly convex.

Sometimes when analyzing Bregman proximal methods, it can also be useful to assume the kernel function $\varphi$ is strongly convex with respect to a norm (we will use this to analyze the PDHG algorithm in Section~\ref{sec:pdhg}). We can show that all kernel functions introduced in Examples~\cref{exmp:cre,exmp:qre,exmp:cld,exmp:qld} are strongly convex over either the probability simplex or set of density matrices.
\begin{prop}\label{prop:entr-strong}
    Negative Shannon and von Neumann entropy are $1$-strongly convex with respect to the $1$- and trace-norm over $\Delta_m$ and $\mathcal{D}(\mathcal{H})$, respectively. 
\end{prop}
\begin{proof}
    These are straightforward consequences of the classical and quantum Pinsker inequalities~\cite[Theorems 10.8.1 and 11.9.1]{wilde2017quantum}.
\end{proof}
\begin{prop}\label{prop:log-strong}
    Negative Burg entropy and the negative log determinant function are $1$-strongly convex with respect to the Euclidean and Frobenius norm over $\Delta_m$ and $\mathcal{D}(\mathcal{H})$, respectively. 
\end{prop}
\begin{proof}
    See Appendix~\ref{sec:log-strong}.
\end{proof}

\section{Quantum Blahut-Arimoto Algorithms}

Blahut-Arimoto algorithms are a type of alternating optimization method which are tailored towards computing channel capacities and related quantities. We define the class of alternating optimization algorithms below.

\begin{defn}[Alternating optimization]
    Consider a bivariate function $g:\mathbb{V}\times\mathbb{V}'\rightarrow \bar{\mathbb{R}}$, and convex sets $\mathcal{C}\subseteq\mathbb{V}$ and $\mathcal{C}'\subseteq\mathbb{V}'$, which define the following bivariate optimization problem
    \begin{equation}\label{eqn:bivariate}
        \min_{x\in\mathcal{C}} \min_{y\in\mathcal{C}'}\,g(x,y).
    \end{equation}
    The alternating optimization algorithm proceeds by using the alternating iterates
    \begin{subequations}\label{eqn:ba}
        \begin{align}
            y^{k+1} &= \argmin_{y\in\mathcal{C}'}\,g(x^k,y), \label{eqn:ba-a}\\
            x^{k+1} &= \argmin_{x\in\mathcal{C}}\,g(x,y^{k+1}).  \label{eqn:ba-b}
        \end{align}
    \end{subequations}    
\end{defn}

To apply the alternating optimization algorithm on constrained convex optimization problems of the form~\eqref{eqn:constr-min}, Blahut-Arimoto algorithms first find a way to express the objective function as a variational expression, i.e., by defining a bivariate extension $g:\mathbb{V}\times\mathbb{V}'\rightarrow \bar{\mathbb{R}}$ of $f$ such that
\begin{equation}\label{eqn:variational}
    f(x) = \min_{y\in\mathcal{C}'}\,g(x,y),
\end{equation}
where $\mathcal{C}'\subseteq\mathbb{V}'$ is a convex set. Notably, the original Blahut-Arimoto algorithm to compute classical channel capacities and rate-distortion functions was derived by noticing that mutual information can be defined using a variational expression of the form~\eqref{eqn:variational}. Given this, the original problem~\eqref{eqn:constr-min} is equivalent to the bivariate optimization problem~\eqref{eqn:bivariate}, on which we can apply the alternating optimization algorithm. This method is practical when minimizing over $x$ and $y$ independently, as described in~\eqref{eqn:ba}, is more efficient than directly solving the original problem in~\eqref{eqn:constr-min}.

We now introduce a particular Blahut-Arimoto framework by~\cite{ramakrishnan2020computing}, which was designed to compute quantum channel capacities.

\begin{defn}[Blahut-Arimoto algorithm~\cite{ramakrishnan2020computing}]\label{def:qba}
    The Blahut-Arimoto algorithm is a specific instance of the alternating optimization algorithm where $\mathbb{V}=\mathbb{V}'=\mathcal{B}(\mathcal{H})$, $\mathcal{C}=\mathcal{C}'=\mathcal{D}(\mathcal{H})$, and \begin{equation}\label{eqn:rama-bivariate}
        g(x,y) = \inp{x}{\mathcal{F}(y)} + \gamma S\divx{x}{y},
    \end{equation}
    for some constant $\gamma>0$ and continuous function $\mathcal{F}:\mathcal{D}(\mathcal{H})_{++}\rightarrow\mathcal{B}(\mathcal{H})$ which satisfies
    \begin{equation}\label{eqn:rama-mono-a}
        0 \leq \inp{x}{\mathcal{F}(x) - \mathcal{F}(y)} \leq \gamma S\divx{x}{y}, \quad \forall x,y\in\mathcal{D}(\mathcal{H})_{++}.
    \end{equation}
\end{defn}

We note that this framework generalizes all other quantum Blahut-Arimoto algorithms~\cite{nagaoka1998algorithms,osawa2001numerical,li2019blahut,li2019computing}. In the next example, we will also show how this framework generalizes the original classical Blahut-Arimoto algorithms~\cite{arimoto1972algorithm,blahut1972computation} by restricting to diagonal systems (see also~\cite{matz2004information,naja2009geometrical}).

\begin{exmp}
    For a classical channel $Q\in\mathcal{Q}_{n, m}$, recall that the classical channel capacity is given by
    \begin{equation*}
        \maximize_{p\,\in \Delta_m} \quad I_\mathrm{c}(p),
    \end{equation*}
    where $I_\mathrm{c}:\mathbb{R}^m_+\rightarrow\mathbb{R}$ is the classical mutual information, and defined as
    \begin{equation}\label{eqn:cmi}
        I_\mathrm{c}(p)\coloneqq \sum_{j=1}^m p_j H\divx{Q_j}{Qp}.
    \end{equation}
    The classical Blahut-Arimoto algorithm~\cite{blahut1972computation,arimoto1972algorithm} to compute this quantity is a specific instance of the alternating optimization algorithm where $\mathbb{V}=\mathbb{V}'=\mathbb{R}^m$, $\mathcal{C}=\mathcal{C}'=\Delta_m$, and
    \begin{equation*}
        g(x,y) = \inp{x}{\mathcal{F}(y)} + H\divx{x}{y},
    \end{equation*}
    where
    \begin{equation*}
        \mathcal{F}(x) = -(H\divx{Q_1}{Qx}, \ldots, H\divx{Q_m}{Qx}).
    \end{equation*}
    Notably, for this choice of $\mathcal{F}$, the expression $x\mapsto -\inp{x}{\mathcal{F}(x)}$ recovers the classical mutual information $I_\mathrm{c}$. Moreover, the gradient of classical mutual information is $\nabla I_\mathrm{c}(x) = -\mathcal{F}(x) + \bm{1}$. This relationship between $\mathcal{F}$ and the gradient of the objective function is not coincidental. We explore this and its implications in Section~\ref{subsec:equiv}.
\end{exmp}

We now recall the convergence analysis associated with the Blahut-Arimoto framework by~\cite{ramakrishnan2020computing}.

\begin{prop}[{\cite[Lemma 3.2, Theorem 3.3, and Proposition 3.6]{ramakrishnan2020computing}}]\label{prop:ba-conv}
    Consider a Blahut-Arimoto algorithm parameterized by constant $\gamma>0$ and continuous function $\mathcal{F}:\mathcal{D}(\mathcal{H})_{++}\rightarrow\mathcal{B}(\mathcal{H})$. This algorithm solves constrained convex optimization problems~\eqref{eqn:constr-min} with objective function
    \begin{equation}\label{eqn:low-dim-f}
        f:\mathcal{D}(\mathcal{H})_{++}\rightarrow\mathbb{R}, \quad f(x)\coloneqq\inp{x}{\mathcal{F}(x)},
    \end{equation}
    and constraint set $\mathcal{C}=\mathcal{D}(\mathcal{H})$, and has iterates~\eqref{eqn:ba} of the form%
    \begin{subequations}\label{eqn:rama-updates}
        \begin{align}
            y^{k+1} &= x^k \label{eqn:rama-updates-a} \\
            x^{k+1} &= \frac{\exp(\log(y^{k+1}) - \mathcal{F}(y^{k+1})/\gamma)}{\tr[\exp(\log(y^{k+1}) - \mathcal{F}(y^{k+1})/\gamma)]}.  \label{eqn:rama-updates-b}
        \end{align}
    \end{subequations}
    These iterates produce a sequence of function values $\{ g(x^{k}, y^{k}) \}$ which converge sublinearly to the global optimum. If $\mathcal{F}$ additionally satisfies
    \begin{equation}\label{eqn:rama-mono-b}
        \inp{x}{\mathcal{F}(x) - \mathcal{F}(y)} \geq a S\divx{x}{y}, \quad \forall x,y\in\mathcal{D}(\mathcal{H})_{++},
    \end{equation}
    for some $a>0$, then the sequence of function values converge linearly.
\end{prop}

\begin{rem}
    The expression~\eqref{eqn:rama-updates-b} is derived without requiring any assumptions on $\mathcal{F}$. Only the second inequality in~\eqref{eqn:rama-mono-a} is required to establish~\eqref{eqn:rama-updates-a} and the form of the objective function being solved.
\end{rem}

\begin{rem}
    Although we assume $\mathcal{F}:\mathcal{D}(\mathcal{H})_{++}\rightarrow\mathcal{B}(\mathcal{H})$ is only defined on the interior of the set of density matrices, it is often possible to extend inner products with $\mathcal{F}$, e.g., $\inp{x}{\mathcal{F}(x)}$, to the boundary of the set by taking the appropriate limit. For the functions we study in this paper, this usually manifests as the convention $0\log(0)=0$ discussed in Section~\ref{subsec:entropy}. This allows us to extend most of the results in this section to the entire set of density matrices.
\end{rem}

\subsection{Equivalence with Mirror Descent}\label{subsec:equiv}

The astute reader may already recognize many similarities between the Blahut-Arimoto algorithm and mirror descent. However, it is not immediately clear if the two algorithms are identical, or if one algorithm is a generalization of the other. In the following theorem, we show that the Blahut-Arimoto framework is in fact a special instance of mirror descent.

\begin{thm}\label{thm:ba-md}
    Consider the data of a Blahut-Arimoto algorithm, i.e., a continuous function $\mathcal{F}:\mathcal{D}(\mathcal{H})_{++}\rightarrow\mathcal{B}(\mathcal{H})$ satisfying~\eqref{eqn:rama-mono-a} for some $\gamma>0$, i.e.,
    \begin{equation*}
        0 \leq \inp{x}{\mathcal{F}(x) - \mathcal{F}(y)} \leq \gamma S\divx{x}{y},
    \end{equation*}
    for all $x,y\in\mathcal{D}(\mathcal{H})_{++}$. Let us define the function
    \begin{equation}\label{eqn:full-f}
        f:\mathcal{B}(\mathcal{H})_{++}\rightarrow\mathbb{R}, \quad f(x) \coloneqq \biggl\langle x,\,\mathcal{F}\biggl(\frac{x}{\tr[x]}\biggr) \biggr\rangle.
    \end{equation}
    Then
    \begin{enumerate}[label=(\alph*)]
        \item $f$ is a convex function which is differentiable on $\mathcal{B}(\mathcal{H})_{++}$, and is $\gamma$-smooth relative to $-S$ on $\mathcal{D}(\mathcal{H})$. If, additionally, $\mathcal{F}$ satisfies~\eqref{eqn:rama-mono-b} for some $a>0$, i.e.,
        \begin{equation*}
            \inp{x}{\mathcal{F}(x) - \mathcal{F}(y)} \geq a S\divx{x}{y},
        \end{equation*}
        for all $x,y\in\mathcal{D}(\mathcal{H})_{++}$, then $f$ is also $a$-strongly convex relative to $-S$ on $\mathcal{D}(\mathcal{H})$.
        \item When initialized at the same point $x^0\in\mathcal{D}(\mathcal{H})_{++}$, the Blahut-Arimoto iterates~\eqref{eqn:rama-updates}, parameterized by $\mathcal{F}$ and $\gamma$, are identical to entropic mirror descent iterates~\eqref{eqn:mtx-exp-grad} with step size $t_k=1/\gamma$ for all $k$, applied to minimizing the convex function $f$ over $\mathcal{D}(\mathcal{H})$.
    \end{enumerate}
\end{thm}

We will first make a few comments about this theorem and its implications, before presenting its proof at the end of this section.

\begin{rem}
    From Proposition~\ref{prop:ba-conv}, the sublinear and linear convergence rates of Blahut-Arimoto algorithms follow from the conditions~\eqref{eqn:rama-mono-a} and~\eqref{eqn:rama-mono-b}, respectively. Similarly, from Proposition~\ref{prop:conv-rate}, the sublinear and linear convergence rates of mirror descent follow from relative smoothness and strong convexity of $f$, respectively. Theorem~\ref{thm:ba-md} tells us that these convergence results are, in fact, equivalent.
\end{rem}

The function~\eqref{eqn:full-f} can be interpreted as a well-defined extension of~\eqref{eqn:low-dim-f} from the affine space $\mathcal{D}(\mathcal{H})$ onto the full-dimensional vector space $\mathcal{B}(\mathcal{H})$, which is required for the objective function to have a well-defined gradient, and, as a result, for mirror descent to be applicable. Later in Proposition~\ref{prop:ba-md-a}, we will show that the gradient of~\eqref{eqn:full-f} is equal to $\mathcal{F}(x)$ for all $x\in\mathcal{D}(\mathcal{H})_{++}$. This is the key property in establishing equivalence between Blahut-Arimoto and mirror descent iterates, which one can see from a direct comparison of~\eqref{eqn:mtx-exp-grad} and~\eqref{eqn:rama-updates}.

\begin{rem}\label{rem:alt-f}
    We note that this extension~\eqref{eqn:full-f} of~\eqref{eqn:low-dim-f} is not unique, and in fact that entropic mirror descent applied to any function of the form
    \begin{equation*}
        \hat{f}:\mathcal{B}(\mathcal{H})_{++}\rightarrow\mathbb{R}, \quad \hat{f}(x) \coloneqq \biggl\langle x,\,\mathcal{F}\biggl(\frac{x}{\tr[x]}\biggr) \biggr\rangle + g(\tr[x]),
    \end{equation*}
    where $g:\mathbb{R}_{++}\rightarrow\mathbb{R}$ is any convex, differentiable function, yields the same sequence of iterates, and satisfies the same relative smoothness and strong convexity properties on $\mathcal{D}(\mathcal{H})$. These functions satisfy $\nabla \hat{f}(x) = \mathcal{F}(x) + g'(1)\mathbb{I}$ for all $x\in\mathcal{D}(\mathcal{H})_{++}$. Notably, this identity component does not affect entropic mirror descent iterates nor its associated convergence analysis.
\end{rem}

Overall, Theorem~\ref{thm:ba-md} tells us that Blahut-Arimoto algorithms are a special case of mirror descent in two regards. First, they are a particular instance of mirror descent using Shannon or von Neumann entropy as the kernel function $\varphi$, and optimize over the set of probability distributions or density matrices. Second, the objective functions of problems solved by Blahut-Arimoto algorithms are of the form~\eqref{eqn:full-f}, or are closely related to this form (see Remark~\ref{rem:alt-f}).

A converse implication to Theorem~\ref{thm:ba-md} is that, by generalizing $-S$ and $\mathcal{D}(\mathcal{H})$ to any arbitrary Legendre function and convex set, respectively, mirror descent algorithms on relatively smooth objective functions of the form~\eqref{eqn:full-f} admit an interpretation as a Blahut-Arimoto-like alternating optimization algorithm. However, it is unclear whether this provides any additional insight in solving certain problems over the more powerful mirror descent generalization.

We now introduce the main ideas we use for the proof of Theorem~\ref{thm:ba-md}, before presenting its proof. We begin by recalling the subdifferential, followed by two results which are used to characterize the gradient of our desired objective function~\eqref{eqn:full-f}.
\begin{defn}[Subdifferential]
    For a convex function $f:\mathbb{V}\rightarrow\bar{\mathbb{R}}$, the \emph{subdifferential} of $f$ at $x\in\domain f$ is
    \begin{equation*}
        \partial f(x)\coloneqq\{ g\in\mathbb{V} : f(y) \geq f(x) + \inp{g}{y-x}, \forall y\in\domain f \}.
    \end{equation*}
    An element of the subdifferential $g\in\partial f(x)$ is called a \emph{subgradient} of $f$ at $x$.
\end{defn}
\begin{lem}\label{lem:ba-md-b}
    Consider a convex function $f:\mathbb{V}\rightarrow\bar{\mathbb{R}}$. If there exists a single-valued continuous operator $\mathcal{G}:\interior\domain f\rightarrow\mathbb{V}$ such that $\mathcal{G}(x)\in\partial f(x)$ for all $x\in\interior\domain f$, then $f$ is differentiable on $\interior\domain f$ with $\nabla f(x)=\mathcal{G}(x)$ for all $x\in\interior\domain f$.
\end{lem}
\begin{proof}
    We will prove the result by contradiction. Assume that there exists $x\in\interior\domain f$ such that $\partial f(x)$ is set-valued. Then there exists a $g\in\partial f(x)$ such that $\norm{g-\mathcal{G}(x)}\coloneqq\varepsilon>0$. As $\mathcal{G}$ is continuous, there exists a $t>0$ and $z=x+t(g-\mathcal{G}(x))$ such that $z\in\interior\domain f$ and $\norm{\mathcal{G}(z) - \mathcal{G}(x)} < \varepsilon$. We can show that
    \begin{align*}
        \norm{\mathcal{G}(z)-\mathcal{G}(x)} &\geq \frac{1}{\norm{z-x}}\inp{z-x}{\mathcal{G}(z)-\mathcal{G}(x)}\\
        &=\frac{1}{\norm{z-x}} (\inp{z-x}{\mathcal{G}(z) - g} + \inp{z-x}{g - \mathcal{G}(x)})\\
        &\geq \frac{1}{\norm{z-x}} \inp{z-x}{g-\mathcal{G}(x)}\\
        &= \norm{g-\mathcal{G}(x)}\\
        &=\varepsilon.
    \end{align*}
    The first inequality uses the Cauchy–Schwarz inequality, the second inequality uses monotonicity of subdifferentials (see, e.g., \cite{ryu2016primer}), and the second equality uses the fact that the two vectors are positive scalar multiples of each other. 
    
    However, this contradicts our assumption that $z$ was chosen to satisfy $\norm{\mathcal{G}(z) - \mathcal{G}(x)} < \varepsilon$. Therefore, the subdifferential $\partial f$ must be single-valued for all $x\in\interior\domain f$, from which the desired result follows from~\cite[Theorem 25.1]{rockafellar1970convex}.
\end{proof}

\begin{prop}\label{prop:ba-md-a}
    Let $\mathcal{F}:\mathcal{D}(\mathcal{H})_{++}\rightarrow\mathcal{B}(\mathcal{H})$ be a single-valued continuous operator which satisfies
    \begin{equation}\label{eqn:rama-f-cal}
        \inp{y}{\mathcal{F}(y)-\mathcal{F}(x)} \geq 0, \quad \forall x,y\in\mathcal{D}(\mathcal{H})_{++}.
    \end{equation}
    Then the function $f:\mathcal{B}(\mathcal{H})_{++}\rightarrow\mathbb{R}, f(x) = \inp{x}{\mathcal{F}(x/\tr[x])} $ is convex, differentiable on $\mathcal{B}(\mathcal{H})_{++}$, and has gradient $\nabla f(x) = \mathcal{F}(x/\tr[x])$ for all $x\in\mathcal{B}(\mathcal{H})_{++}$.
\end{prop}
\begin{proof}%Proof of Proposition~\ref{prop:ba-md-a}]
    Using~\eqref{eqn:rama-f-cal},  for all $x,y\in\mathcal{B}(\mathcal{H})_{++}$ we have that
    \begin{align}
        f(y) = \biggl\langle y,\,\mathcal{F}\biggl(\frac{y}{\tr[y]}\biggr) \biggr\rangle & \geq \biggl\langle y,\,\mathcal{F}\biggl(\frac{x}{\tr[x]}\biggr) \biggr\rangle \nonumber\\
        &= \biggl\langle x,\,\mathcal{F}\biggl(\frac{x}{\tr[x]}\biggr) \biggr\rangle + \biggl\langle y-x,\,\mathcal{F}\biggl(\frac{x}{\tr[x]}\biggr) \biggr\rangle \nonumber\\
        &= f(x) + \biggl\langle y-x,\,\mathcal{F}\biggl(\frac{x}{\tr[x]}\biggr) \biggr\rangle. \label{eqn:rama-proof-c}
    \end{align}
    This implies, by definition, that $\mathcal{F}(x/\tr[x])\in\partial f(x)$ for all $x\in\mathcal{B}(\mathcal{H})_{++}$. Convexity of $f$ follows from the existence of a subgradient for all $x\in\mathcal{B}(\mathcal{H})_{++}$. To show this, let $\bar{z}=\lambda \bar{x} + (1-\lambda) \bar{y}$ for some $\bar{x},\bar{y}\in\mathcal{B}(\mathcal{H})_{++}$ and $\lambda\in[0,1]$. Then using~\eqref{eqn:rama-proof-c} we can show that
    \begin{align*}
        & f(\bar{x}) \geq f(\bar{z}) + \biggl\langle \bar{x}-\bar{z},\,\mathcal{F}\biggl(\frac{\bar{z}}{\tr[\bar{z}]}\biggr) \biggr\rangle = f(\bar{z}) + (1 - \lambda) \biggl\langle \bar{x} - \bar{y},\,\mathcal{F}\biggl(\frac{\bar{z}}{\tr[\bar{z}]}\biggr) \biggr\rangle \\
        \textrm{and} \qquad &f(\bar{y}) \geq f(\bar{z}) + \biggl\langle \bar{y}-\bar{z},\,\mathcal{F}\biggl(\frac{\bar{z}}{\tr[\bar{z}]}\biggr) \biggr\rangle = f(\bar{z}) - \lambda \biggl\langle \bar{x} - \bar{y},\,\mathcal{F}\biggl(\frac{\bar{z}}{\tr[\bar{z}]}\biggr) \biggr\rangle,
    \end{align*}
    where we have used $\bar{x}-\bar{z}=(1-\lambda)(\bar{x}-\bar{y})$ and $\bar{y}-\bar{z}=\lambda(\bar{y}-\bar{x})$ for the equalities. Combining $\lambda$ times the first inequality and $(1-\lambda)$ times the second inequality gives $\lambda f(\bar{x}) + (1-\lambda)f(\bar{y}) \geq f(\bar{z})$, which proves convexity of $f$. From continuity of $\mathcal{F}$ and Lemma~\ref{lem:ba-md-b}, we obtain differentiability of $f$ with $\nabla f(x) = \mathcal{F}(x/\tr[x])$ for all $x\in\mathcal{B}(\mathcal{H})_{++}$, which concludes the proof.
\end{proof}

We now show that the assumptions~\eqref{eqn:rama-mono-a} and~\eqref{eqn:rama-mono-b} are identical to relative smoothness and strong convexity for functions satisfying a homogeneity-like property.

\begin{prop}\label{cor:mono-rel}
    Consider a function $f:\mathbb{V}\rightarrow\bar{\mathbb{R}}$ where $f$ is differentiable on $\interior\domain f$, a Legendre function $\varphi:\mathbb{V}\rightarrow\bar{\mathbb{R}}$, and a convex set $\mathcal{C}\subseteq\closure\domain f \cap \closure\domain\varphi$. Assume that $f$ satisfies the identity
    \begin{equation}\label{eqn:homo}
        f(x) = \inp{x}{\nabla f(x)}, \quad \forall x\in\relinterior\mathcal{C}.
    \end{equation}
    Then $f$ is $L$-smooth relative to $\varphi$ on $\mathcal{C}$ if and only if
    \begin{equation} \label{eqn:one-sided-l-smooth}
        \inp{x}{\nabla f(x) - \nabla f(y)} \leq L D_\varphi\divx{x}{y}, \quad \forall x,y\in\relinterior\mathcal{C}.
    \end{equation}
    Similarly, $f$ is $\mu$-strongly convex relative to $\varphi$ on $\mathcal{C}$ if and only if
    \begin{equation} \label{eqn:one-sided-l-strong}
        \inp{x}{\nabla f(x) - \nabla f(y)} \geq \mu D_\varphi\divx{x}{y}, \quad \forall x,y\in\relinterior\mathcal{C}.
    \end{equation}
\end{prop}
\begin{proof}
    First, assume $f$ is $L$-smooth relative to $\varphi$.  Then for all $x,y\in\relinterior\mathcal{C}$
    \begin{align*}
        \inp{x}{\nabla f(x) - \nabla f(y)} &= f(x) - f(y) - \inp{\nabla f(y)}{x - y}\\
        &\leq LD_\varphi\divx{x}{y},
    \end{align*}
    where the equality uses~\eqref{eqn:homo}, and the inequality uses Proposition~\ref{prop:relative-alt}(a-iii). This recovers~\eqref{eqn:one-sided-l-smooth} as desired. To show the converse, we can simply add together a copy of~\eqref{eqn:one-sided-l-smooth} with another copy of itself with the roles of $x$ and $y$ exchanged to obtain 
    \begin{equation*}
        \inp{x - y}{\nabla f(x) - \nabla f(y)} \leq L (D_\varphi\divx{x}{y} + D_\varphi\divx{x}{y}),
    \end{equation*}
    for any $x,y\in\relinterior\mathcal{C}$. This is precisely the same expression as Proposition~\ref{prop:relative-alt}(a-ii), and therefore implies relative smoothness of $f$, as desired. The proof for relative strong convexity follows from a similar argument. 
\end{proof}
\begin{rem}\label{rem:mono}
    The proof that~\eqref{eqn:one-sided-l-smooth} and~\eqref{eqn:one-sided-l-strong} imply relative smoothness and strong convexity, respectively, do not require the function to satisfy~\eqref{eqn:homo}, and is true for any function with a suitable domain. Therefore, \eqref{eqn:one-sided-l-smooth} and~\eqref{eqn:one-sided-l-strong} are in general stronger assumptions than relative smoothness and strong convexity. For example, \eqref{eqn:one-sided-l-strong} does not hold when $f(x)=\varphi(x)=\norm{x}_2^2/2$ and $L=1$, even though $f$ is $L$-smooth relative to $\varphi$.
\end{rem}

We now provide the proof for Theorem~\ref{thm:ba-md}.

\begin{proof}[Proof of Theorem~\ref{thm:ba-md}]
    For part (a), convexity and differentiability of $f$ follow from Proposition~\ref{prop:ba-md-a}. The same proposition also gives us the identity $\nabla f(x) = \mathcal{F}(x)$ for all $x\in\mathcal{D}(\mathcal{H})_{++}$. Combining this with Proposition~\ref{cor:mono-rel} gives us the desired relative smoothness and strong convexity properties. Part (b) follows from using the same identity to comparing between the Blahut-Arimoto updates~\eqref{eqn:rama-updates} and entropic mirror descent~\eqref{eqn:mtx-exp-grad} with step size $t_k=1/\gamma$ for all $k$.
\end{proof}

\section{Primal-Dual Hybrid Gradient} \label{sec:pdhg}

When computing energy-constrained channel capacities or related quantities such as rate-distortion functions, we are required to optimize over the intersection of the probability simplex (density matrices) together with additional linear constraints. These problems can be expressed as the following variation of the constrained convex optimization problem~\eqref{eqn:constr-min}
\begin{subequations}\label{eqn:constr-constr-min}
    \begin{align}
        \minimize_{x\in\mathcal{C}} \quad & f(x) \\
        \subjto \quad & b_1 - \mathcal{A}_1(x) \in \mathcal{K} \label{eqn:constr-constr-min-b}\\
                \quad &  b_2 - \mathcal{A}_2(x) = 0, \label{eqn:constr-constr-min-c}
    \end{align}
\end{subequations}
where $\mathcal{K}\subseteq\mathbb{V}_1$ is a proper cone (see, e.g., \cite[Section 2.4]{boyd2004convex}), $b_1\in\mathbb{V}_1$, $b_2\in\mathbb{V}_2$, and $\mathcal{A}_1:\mathbb{V}\rightarrow\mathbb{V}_1$ and $\mathcal{A}_2:\mathbb{V}\rightarrow\mathbb{V}_2$ are linear operators. Blahut-Arimoto algorithms have traditionally been unable to effectively handle these additional constraints, and typically dualize the constraints to avoid the issue. One possible method of solving~\eqref{eqn:constr-constr-min} is to use mirror descent. However, recall that mirror descent is only practical when the updates~\eqref{eqn:mirror-descent} can be efficiently computed. For example, for the kernel function $\varphi=-H$, the mirror descent step no longer has the analytic expression~\eqref{eqn:exp-grad-desc} for general linear constraints, and must be computed numerically. %We remark that numerically solving the dual of this subproblem leads to a similar algorithm as~\cite{hayashi2022bregman}.

Alternatively, instead of a primal first-order method, we can consider using a primal-dual first-order method instead. These methods find the primal-dual pair $(x^*, z^*)$ which solves the saddle-point problem
\begin{equation}\label{eqn:saddle-point}
    \newinf_{x\in\mathcal{C}} \sup_{z\in\mathcal{Z}} \quad \mathcal{L}(x, z) \coloneqq f(x) + \inp{z}{\mathcal{A}(x) - b},
\end{equation}
where $\mathcal{L}$ is the Lagrangian and $z$ is the dual variable of~\eqref{eqn:constr-constr-min}, $\mathcal{Z}=\{(\lambda, \nu)\in\mathcal{K}^*\times\mathbb{V}_2\}$ where $\mathcal{K}^*$ is the dual cone of $\mathcal{K}$, and we have defined the concatenated operator $\mathcal{A}=(\mathcal{A}_1, \mathcal{A}_2):\mathbb{V} \rightarrow \mathbb{V}_1 \times \mathbb{V}_2$ and vector $b=(b_1, b_2)\in\mathbb{V}_1 \times \mathbb{V}_2$ to simplify notation. Solving the saddle-point problem~\eqref{eqn:saddle-point} is equivalent to solving the original problem~\eqref{eqn:constr-constr-min} in the sense that $\mathcal{L}(x^*, z^*)=f(x^*)=f^*$, where $f^*$ is the optimal value of~\eqref{eqn:constr-constr-min}. 

An extension of Bregman proximal methods to solve these types of saddle-point problems is PDHG~\cite{pock2009algorithm,chambolle2011first,chambolle2016ergodic,jiang2022bregman}. There are several variations of PDHG. In Algorithm~\ref{alg:pdhg}, we introduce a slight modification of the main algorithms presented in~\cite{chambolle2016ergodic,jiang2022bregman}.

\begin{algorithm*}
    \caption{Primal-dual hybrid gradient}
    \begin{algorithmic}
        \Input{Objective function $f$, Legendre kernel function $\varphi$, primal domain $\mathcal{C}\subseteq\domain\varphi$,}
        \Indent{ linear constraint parameters $\mathcal{A}, b$, step sizes $\theta_k,\tau_k,\gamma_k>0$}
        \State \textbf{Initialize:} Initial points $x^0\in\relinterior\mathcal{C}$, $z^0,z^{-1}\in\mathcal{Z}$.
        \For{$k=0,1,\ldots$}{}
        \begin{subequations}
            \begin{align}
                &\Bar{z}^{k+1}= z^k + \theta_k(z^k - z^{k-1})\label{eqn:pdhg-a}\\
                &\displaystyle x^{k+1} = \argmin_{x \in \mathcal{C}} \mleft\{ \inp{\nabla f(x) + \mathcal{A}^\dag (\Bar{z}^{k+1})}{x} + \frac{1}{\tau_k} D_\varphi\divx{x}{x^k} \mright\}\\
                &\displaystyle z^{k+1} = \argmin_{z \in \mathcal{Z}} \biggl\{ -\inp{z}{\mathcal{A}(x^{k+1}) - b} + \frac{1}{2\gamma_k} \norm{z - z^k}_2^2 \biggr\}
            \end{align}
        \end{subequations}
        \EndFor
        \Output{Approximate primal-dual solution $(x^{k+1}, z^{k+1})$}
    \end{algorithmic}
    \label{alg:pdhg}
\end{algorithm*}

We make a few remarks about the PDHG algorithm. First, note that the primal update is precisely the mirror descent update~\eqref{eqn:mirror-descent} performed on the Lagrangian, and the dual update is a mirror ascent update on the Lagrangian using the energy function as the kernel function (see Example~\ref{exmp:energy}), which recovers projected gradient descent. This dual step is efficient whenever Euclidean projection onto the dual cone $\mathcal{K}^*$ can be done efficiently. This is the case for the non-negative orthant $\mathcal{K}=\mathcal{K}^*\coloneqq\mathbb{R}_+^n$ which yields
\begin{align}
    \proj_{\mathbb{R}_+^n}(x) = x_+ \coloneqq \max(x, 0),
\end{align}
or the positive semidefinite cone $\mathcal{K}=\mathcal{K}^*\coloneqq\mathbb{H}_+^n$
\begin{align}
    \proj_{\mathbb{H}_+^n}(X) = X_+ \coloneqq U\max(\Lambda, 0)U^\dag,
\end{align}
where $X\in\mathbb{H}^n$ has diagonalization $X=U\Lambda U^\dag$, and we have defined $\max(\wc, \wc)$ to be taken elementwise. Second, if we choose $\varphi(\wc)=\norm{\wc}^2_2/2$ and $\theta_k=0$ for all $k$, then we recover the Arrow-Hurwicz method~\cite{arrow1958studies}. Although this algorithm may seem more natural, as we no longer require the seemingly unintuitive step~\eqref{eqn:pdhg-a}, the Arrow-Hurwicz method is only known to converge under fairly restrictive step size assumptions (see, e.g.,~\cite{he2014convergence}), and it is not obvious how to extend these convergence results for the Arrow-Hurwicz method to the Bregman proximal setting. The step~\eqref{eqn:pdhg-a} is an appropriate modification to the algorithm to obtain more desirable convergence behaviours, which we will introduce in Proposition~\ref{prop:pdhg-conv}. Third, constraints are handled in two different ways. Constraints which can be efficiently minimized over in the mirror descent update are encoded in $\mathcal{C}$, whereas all other constraints are numerically handled through the dual updates. 

Like mirror descent, relative smoothness is a core assumption which we use to establish convergence of Algorithm~\ref{alg:pdhg}. Using this, we are able to establish the following ergodic sublinear convergence rate, which is the standard basic convergence result for PDHG algorithms without further assumptions~\cite{chambolle2016ergodic,jiang2022bregman}.

\begin{prop}\label{prop:pdhg-conv}
    Consider Algorithm~\ref{alg:pdhg} to solve the convex optimization problem~\eqref{eqn:constr-constr-min}. If $f$ is $L$-smooth relative to $\varphi$, $\theta_k=1$ for all $k$, and $\tau_k=\tau$ and $\gamma_k=\gamma$ are chosen such that
    \begin{equation}\label{eqn:pdhg-cond}
        \mleft( \frac{1}{\tau} - L \mright) D_\varphi\divx{x}{x'} + \frac{1}{2\gamma} \norm{z - z'}_2^2 \geq \inp{z - z'}{\mathcal{A}(x - x')},
    \end{equation}
    for all $x,x'\in\mathcal{C}$ and $z,z'\in\mathcal{Z}$, then the iterates $(x^k, z^k)$ satisfy
    \begin{align}\label{eqn:pdhg-gap}
        \mathcal{L}(x_{\mathrm{avg}}^k, z) - \mathcal{L}(x, z_{\mathrm{avg}}^k) \leq \frac{1}{k} \mleft( \frac{1}{\tau} D_\varphi\divx{x}{x^0} + \frac{1}{2\gamma}\norm{z - z^0}^2_2 \mright),
    \end{align}
    for all $x\in\mathcal{C}$, $z\in\mathcal{Z}$, and $k\in\mathbb{N}$, where
    \begin{equation}
        x_{\mathrm{avg}}^k = \frac{1}{k} \sum_{i=1}^{k} x^i \quad \textrm{and} \quad z_{\mathrm{avg}}^k = \frac{1}{k} \sum_{i=1}^{k} \Bar{z}^i.
    \end{equation}
\end{prop}
\begin{proof}
    We refer to the proof for Proposition~\ref{prop:appdx-pdhg-conv} (in Appendix~\ref{appdx:bcktrack-pdhg}), which provides convergence rates for a generalized algorithm.
\end{proof}

\begin{rem}
    Consider choosing any optimal primal-dual solution $(x, z)=(x^*, z^*)$ for~\eqref{eqn:pdhg-gap}. As $(x^*, z^*)$ is a saddle-point, we know that the left hand side of~\eqref{eqn:pdhg-gap} is non-negative and therefore converges to zero sublinearly. Additionally, this expression equals zero if and only if $(x_{\mathrm{avg}}^k, z_{\mathrm{avg}}^k)$ is itself a saddle-point. We refer the reader to, e.g.,~\cite{chambolle2016ergodic,jiang2022bregman} for additional convergence properties of PDHG under suitable assumptions.
\end{rem}

\begin{rem}\label{ref:pdhg-alt}
    The condition~\eqref{eqn:pdhg-cond} is equivalent to
    \begin{equation}
        \mleft( \frac{1}{\tau} - L \mright) \frac{1}{\gamma} D_\varphi\divx{x}{x'} \geq \frac{1}{2} \norm{\mathcal{A}(x-x')}_2^2,
    \end{equation}
    for all $x,x'\in\mathcal{C}$. In general, it is not clear whether there exist step sizes that satisfy this condition. If however $\varphi$ is $1$-strongly convex with respect to some norm $\norm{\wc}$, then a sufficient condition to satisfy~\eqref{eqn:pdhg-cond} is for the step sizes $\tau$ and $\gamma$ to satisfy
    \begin{equation}\label{eqn:pdhg-cond-alt}
        \mleft( \frac{1}{\tau} - L \mright) \frac{1}{\gamma} \geq \norm{\mathcal{A}}^2,
    \end{equation}
    where
    \begin{equation}
        \norm{\mathcal{A}} = \sup\{ \inp{z}{\mathcal{A}(x)} : \norm{x}\leq1, \norm{z}_2\leq1 \}.
    \end{equation}
    This inequality can always be achieved by choosing $\tau$ and $\gamma$ to be sufficiently small. Therefore, this inequality implies the existence of a suitable step size and provides a way to choose a suitable step size, whenever $\varphi$ is strongly convex with respect to a norm.
\end{rem}

There are a few limitations with the PDHG algorithm. We do not address all of these limitations in this paper, and leave some of them for future work. First, as a primal-dual method, not all primal iterates generated by PDHG are guaranteed to be feasible. Depending on how difficult it is to perform projections into the feasible set, this may make computing an explicit optimality gap difficult. Second, the primal-dual step size ratio has a significant impact on the empirical convergence rates of PDHG. Some works~\cite{malitsky2018first} suggest backtracking methods to adaptively select this ratio, however it is not straightforward how to generalize these heuristics in a way which maintains convergence guarantees. Third, choosing step sizes which satisfy the assumption~\eqref{eqn:pdhg-cond} or~\eqref{eqn:pdhg-cond-alt} can be nontrivial. We detail a backtracking procedure which adaptively chooses these step sizes, and provide convergence results, in Appendix~\ref{appdx:bcktrack-pdhg}.

%% =================================================================================
%% =================================================================================

\section{Applications} \label{sec:app}

In this section, we will show how the mirror descent framework can be used to solve various problems in information theory. We use PDHG with entropic kernel functions to compute energy-constrained channel capacities and rate-distortion functions in Sections~\ref{subsec:cc-pd} and~\ref{sec:rd}, respectively. We use PDHG with a negative log-determinant kernel function to compute the relative entropy of a quantum  resource in Section~\ref{subsec:ree}. These are problems which Blahut-Arimoto algorithms have not been applied to, or are unable to solve efficiently. For each of these applications, we will prove that the objective functions are smooth relative to some suitable kernel function, which subsequently allows us to guarantee convergence rates when applying mirror descent or PDHG when suitable step sizes are used. 

All experiments were run on MATLAB using an Intel i5-11700 CPU with 32GB of RAM. We use the backtracking variation of PDHG introduced in Appendix~\ref{appdx:bcktrack-pdhg}, with backtracking parameters $\alpha=0.75$ and $\Bar{\theta}=1.01$. The ratio of primal and dual step sizes $\kappa\coloneqq\tau_{k}/\gamma_{k}$ was manually tuned for each problem by running PDHG on each problem class across a few values of $\kappa$, and seeing which value resulted in the fastest convergence. These values will be specified for each problem with their results. The PDHG algorithm is terminated when the progress made by the primal and dual iterates slows down below a chosen threshold, as measured by
\begin{equation}\label{eqn:stop-tol}
    \frac{D_\varphi\divx{x^k}{x^{k-1}}}{\tau_k \max\{1, \norm{x^k}_\infty\}} + \frac{\norm{z^k - z^{k-1}}_2^2}{2\gamma_k\max\{1, \norm{z^k}_\infty\}} \leq 10^{-7},
\end{equation}
where $\norm{\wc}_\infty$ represents the maximum absolute value of the elements of the vector or matrix. This heuristic is based on the form~\eqref{eqn:pdhg-gap}, combined with normalizing terms which measure the divergences relative to the magnitudes of the primal and dual variables.

We compare the computation times and absolute tolerances of PDHG against MOSEK~\cite{mosek}, which is a state-of-the-art, off-the-shelf commercial solver which implements an interior-point algorithm. To parse our problems into a form solvable by MOSEK, we use the standard modelling software, CVX~\cite{cvx,gb08}. For quantum problems, CVX uses CVXQUAD~\cite{fawzi2019semidefinite} to approximate matrix logarithms using linear matrix inequalities, which can then be solved as a semidefininte program. We report the total time taken for CVX to model the problem as well as the solve time of MOSEK. Reported optimality gaps are computed either using information from MOSEK, or by running PDHG for a sufficiently high number of iterations if MOSEK fails to solve.

We make a few remarks about how one should interpret the results we present in the following subsections. As PDHG is a first-order method, we are primarily interested in how well we can quickly obtain low- to medium-accuracy solutions for large-scale problems. Therefore, we present results for a range of problem dimensions, and choose a stopping criterion~\eqref{eqn:stop-tol} which achieves a suitable tradeoff between accuracy and computation time. We note that when comparing the computation times between PDHG and MOSEK, that MOSEK implements a second-order method which is known to rapidly converge to high accuracy solutions when sufficiently close to the optimum. Therefore, MOSEK takes roughly the same amount of time to reach medium-accuracy solutions as they do to reach the high-accuracy solutions we report in the results. We also note that we often report that MOSEK fails to solve medium- to large-scale problems due to insufficient computational memory. This is also due to MOSEK using a second-order method, which inherently needs to compute and store much larger matrices compared to first-order methods. Additionally, the semidefinite approximations used by CVXQUAD can have matrix dimensions much larger than the dimensions of the matrix logarithms they approximate. We also remark in advance that some oscillations seen in the convergence plots in Figures~\cref{fig:channel-capacities,fig:rate-distortions,fig:ree} arise from the backtracking procedure in Algorithm~\ref{alg:pdhg-backtrack} slowly increasing the step sizes due to our choice of $\bar{\theta}>1$, until the algorithm detects that it is close to becoming unstable and performs a backtracking step to stabilize the step sizes. Other oscillations are due to the plots showing the absolute optimality gap, and the fact that PDHG does not monotonically converge to the optimal value.

Throughout this section, by using the MATLAB package QETLAB~\cite{qetlab}, random probability distributions are sampled uniformly on the probability simplex, and random density matrices are sampled uniformly according to the Hilbert-Schmidt measure. Random classical channels are generated by independently sampling a random probability distribution for each column of the matrix. Random quantum channels are sampled uniformly according to the Hilbert-Schmidt measure, using the method described by~\cite{kukulski2021generating}.

\subsection{Energy-Constrained Channel Capacities} \label{subsec:cc-pd}

We will consider three types of classical and quantum channel capacities with an arbitrary number of energy (or any other linear) constraints. First, consider the discrete input alphabet $\mathcal{X}=\{ 1, \ldots, m \}$ where each letter is sent according to a probability distribution $p\in\Delta_m$, a classical channel with $m$ inputs and $n$ outputs represented by $Q\in\mathcal{Q}_{n, m}$ such that $Q_{ij}$ is the conditional probability of receiving symbol $i$ given than symbol $j$ was sent, and $l$ energy constraints on the channel represented by $A\in\mathbb{R}^{l \times m}_+$ and $b\in\mathbb{R}^l_+$. The \emph{energy-constrained classical channel capacity} is given by
\begin{subequations}\label{eqn:cc}
    \begin{align}
        \mathmakebox[\widthof{$\subjto$}]{\maximize_{p\,\in \Delta_m}} \quad &I_\mathrm{c}(p) \\
        \subjto \quad & Ap \leq b,
    \end{align}
\end{subequations}
where $I_\mathrm{c}:\mathbb{R}^m_+\rightarrow\mathbb{R}$ is the \emph{classical mutual information}, defined earlier in~\eqref{eqn:cmi}, and has partial derivatives
\begin{equation}
    \frac{\partial I_\mathrm{c}}{\partial p_j} = H\divx{Q_j}{Qp} - 1.
\end{equation}
The classical mutual information can be shown to be concave by expressing it as
\begin{equation}\label{eqn:grad-clas-mut}
    I_\mathrm{c}(p) = H(Qp) - \sum_{j=1}^m p_j H(Q_j),
\end{equation}
then noting that Shannon entropy is concave, and the second term is linear in $p$. As $\Delta_m$ is a convex set and all energy constraints are linear, we conclude that~\eqref{eqn:cc} is a convex optimization problem.

Again consider the discrete input alphabet $\mathcal{X}$ with input distribution $p$ and energy constraints represented by $A$ and $b$, however now we are interested in communicating the input signal through a quantum channel $\mathcal{N} \in \Phi(\mathcal{H}_A, \mathcal{H}_B)$. To do this, each letter is represented by a quantum state $i \mapsto \ket{i}\!\bra{i}$ for an orthonormal basis $\{\ket{i}\}$ defined on $\mathcal{H}_X$. An encoder $\mathcal{E} \in \Phi(\mathcal{H}_X, \mathcal{H}_A)$ is used to map these classical states to a fixed quantum state alphabet $\ket{i}\!\bra{i}\mapsto\rho_i$. A decoder, which can perform joint output measurements, is then used to recover the original message. The maximum information that can be transmitted through this system is known as the \emph{energy-constrained classical-quantum (cq) channel capacity}~\cite{schumacher1997sending, holevo1998capacity}, and is given by
\begin{subequations}\label{eqn:cq}
    \begin{align}
        \mathmakebox[\widthof{$\subjto$}]{\maximize_{p\,\in \Delta_m}} \quad &\chi(p) \\
        \subjto \quad & Ap \leq b,
    \end{align}
\end{subequations}
where $\chi:\mathbb{R}^m_+\rightarrow\mathbb{R}$ is the \emph{Holevo information} defined as
\begin{equation}
    \chi(p) \coloneqq S\biggl(\mathcal{N}\biggl(\sum_{j=1}^m p_j\rho_j\biggr)\biggr) - \sum_{j=1}^m p_jS(\mathcal{N}(\rho_j)).
\end{equation}
To find the partial derivative of $\chi$ with respect to $p_j$, we note that from linearity of the quantum channel $\mathcal{N}$, the derivative of the first term can be found from a straightforward application of Lemma~\ref{lem:dir-deriv}(b) (in Appendix~\ref{appdx:grad}) to find $\mathsf{D}S(\mathcal{N}(\rho))[\mathcal{N}(\rho_j)]$, where $\rho=\sum_{j=1}^mp_j\rho_j$. The second term is linear in $p_j$. Therefore, 
\begin{align}
    \frac{\partial \chi}{\partial p_j} &= -\tr[\mathcal{N}(\rho_j)(\log(\mathcal{N}(\rho)) + \mathbb{I})] - S(\mathcal{N}(\rho_j)) \nonumber\\
    &= \tr[\mathcal{N}(\rho_j)(\log(\mathcal{N}(\rho_j)-\log(\mathcal{N}(\rho)) - \mathbb{I})] \nonumber\\
    &= S\divx{\mathcal{N}(\rho_j)}{\mathcal{N}(\rho)} - 1, \label{eqn:holevo-grad}
\end{align}
where the last equality uses the fact that $\mathcal{N}$ is trace preserving and $\rho_j$ is a density matrix with unit trace. Similar to classical mutual information, Holevo information can be shown to be concave by noting that von Neumann entropy is concave and the second term is linear. As $\Delta_m$ is a convex set and all energy constraints are linear, \eqref{eqn:cq} is a convex optimization problem.

Now consider the case where the quantum state alphabet is not fixed, input states can be entangled, and the sender and receiver share an unlimited number of entangled states prior to sending messages through the channel. Energy constraints are now represented by $l$ positive observables $A_i\in\mathcal{B}(\mathcal{H}_A)_+$ and upper bounds $b_i\geq0$. The maximum information that can be transmitted through this system is the \emph{energy-constrained entanglement-assisted (ea) channel capacity}~\cite{bennett2002entanglement}, and is given by
\begin{subequations}\label{eqn:ea}
    \begin{align}
        \maximize_{\rho\,\in \mathcal{D}(\mathcal{H}_A)} \quad &I_\mathrm{q}(\rho) \\
        \mathmakebox[\widthof{$\maximize_{\rho\,\in \mathcal{D}(\mathcal{H}_A)}$}]{\subjto} & \tr[A_i \rho] \leq b_i, \qquad \textrm{ for } i=1,\ldots,l,
    \end{align}
\end{subequations}
where $I_\mathrm{q}:\mathcal{B}(\mathcal{H}_A)_+\rightarrow\mathbb{R}$ is the \emph{quantum mutual information} defined as
\begin{equation}
    I_\mathrm{q}(\rho) \coloneqq S(\rho) + S(\mathcal{N}(\rho)) - S(\mathcal{N}_{\mathrm{comp}}(\rho)),
\end{equation}
where $\mathcal{N}_{\mathrm{comp}}$ is the complementary channel of $\mathcal{N}$ defined in Section~\ref{secsub:notation}. To find the gradient of quantum mutual information, we again recognize that $\mathcal{N}$ and $\mathcal{N}_{\mathrm{comp}}$ are linear operators, and apply Corollaries~\ref{cor:grad-tr-f} and~\ref{cor:grad-tr-f-lin} (in Appendix~\ref{appdx:grad}) to find
\begin{align}
    \nabla I_\mathrm{q}(\rho) &= -\log(\rho) -\mathcal{N}^\dag(\log(\mathcal{N}(\rho)) + \mathbb{I}) + \mathcal{N}_{\mathrm{comp}}^\dag(\log(\mathcal{N}_{\mathrm{comp}}(\rho)) + \mathbb{I}) - \mathbb{I} \nonumber\\
    &= -\log(\rho) - \mathcal{N}^\dag(\log(\mathcal{N}(\rho))) + \mathcal{N}_{\mathrm{comp}}^\dag(\log(\mathcal{N}_{\mathrm{comp}}(\rho))) - \mathbb{I}, \label{eqn:qmi-grad}
\end{align}  
where the second equality uses the fact that the adjoint of trace preserving linear operators are unital~\cite[Proposition 4.4.1]{wilde2017quantum}. It can be shown that quantum mutual information is concave by expressing it as 
\begin{equation}
    I_\mathrm{q}(\rho) = S\divy{B}{E}_{U\rho U^\dag} + S(\mathcal{N}(\rho)),
\end{equation}
where $U$ is the isometry corresponding to the Stinespring representation $\mathcal{N}$, and noting that both quantum conditional entropy and von Neumann entropy are concave in $\rho$. As $\mathcal{D}(\mathcal{H}_A)$ is a convex set and all energy constraints are linear, it follows that~\eqref{eqn:ea} is also a convex optimization problem.

We will now show that negative classical and quantum mutual information and Holevo information are all relatively smooth with respect to Shannon and von Neumann entropy. We first recall the following channel parameters as introduced in~\cite{matz2004information,ramakrishnan2020computing}.
\begin{defn}[Channel coefficients]\label{def:coeff}
    For a classical channel $Q \in \mathcal{Q}_{n,m}$, the contraction coefficient $\zeta_Q$ is defined as
    \begin{equation}
        \zeta_Q = \sup \mleft\{ \frac{H\divx{Qp}{Qq}}{H\divx{p}{q}} : p,q \in \Delta_m,  p \neq q \mright\},
    \end{equation}
    and the expansion coefficient $\eta_Q$ is defined as
    \begin{equation}   
        \eta_Q = \inf \mleft\{ \frac{H\divx{Qp}{Qq}}{H\divx{p}{q}} : p,q \in \Delta_m,  p \neq q \mright\}.
    \end{equation}
    Similarly, for a quantum channel $\mathcal{N} \in \Phi(\mathcal{H}_A, \mathcal{H}_B)$, the contraction coefficient $\zeta_\mathcal{N}$ is defined as
    \begin{equation}
        \hspace{-0.15em}\zeta_\mathcal{N} = \sup \mleft\{ \frac{S\divx{\mathcal{N}(\rho)}{\mathcal{N}(\sigma)}}{S\divx{\rho}{\sigma}} : \rho, \sigma \in \mathcal{D}(\mathcal{H}_A),  \rho \neq \sigma \mright\},
        \label{eqn:contract-coeff}
    \end{equation}
    and the expansion coefficient $\eta_\mathcal{N}$ is defined as
    \begin{equation}   
        \eta_\mathcal{N} = \inf \mleft\{ \frac{S\divx{\mathcal{N}(\rho)}{\mathcal{N}(\sigma)}}{S\divx{\rho}{\sigma}} : \rho, \sigma \in \mathcal{D}(\mathcal{H}_A),  \rho \neq \sigma \mright\}.
    \end{equation}
\end{defn}

We now show how these channel coefficients are related to relative smoothness and strong convexity. 
\begin{thm}\label{thm:cq-ea-rel}
    Consider a classical channel $Q \in \mathcal{Q}_{m,n}$, a quantum channel $\mathcal{N} \in \Phi(\mathcal{H}_A, \mathcal{H}_B)$, and an encoding channel $\mathcal{E} \in \Phi(\mathcal{H}_X, \mathcal{H}_A)$.
    \begin{enumerate}[label=(\alph*)]
        \item Negative classical mutual information $-I_\mathrm{c}$ is $\zeta_Q$-smooth and $\eta_Q$-strongly convex relative to negative Shannon entropy $-H(\wc)$ on $\Delta_m$.
        \item Negative Holevo information $-\chi$ is $\zeta_{(\mathcal{N}\circ\mathcal{E})}$-smooth and $\eta_{(\mathcal{N}\circ\mathcal{E})}$-strongly convex relative to negative Shannon entropy $-H(\wc)$ on $\Delta_m$.
        \item Negative quantum mutual information $-I_\mathrm{q}$ is $(1+\zeta_\mathcal{N}-\eta_{\mathcal{N}_{\mathrm{comp}}})$-smooth and $(1+\eta_\mathcal{N}-\zeta_{\mathcal{N}_{\mathrm{comp}}})$-strongly convex relative to negative von Neumann entropy $-S(\wc)$ on $\mathcal{D}(\mathcal{H}_A)$.
    \end{enumerate}
\end{thm}
\begin{proof}
    To show part (a), using the expression for the gradient of classical mutual information~\eqref{eqn:grad-clas-mut}, it is possible to show that for any probability distributions $p,q\in\Delta_m$ that
    \begin{align*}
        \inp{p}{\nabla I_\mathrm{c}(q) - \nabla I_\mathrm{c}(p)} &= H\divx{Qp}{Qq},
    \end{align*}
    and therefore using the definition of channel coefficients
    \begin{align*}
        \eta_Q H\divx{p}{q} \leq \inp{p}{\nabla I_\mathrm{c}(q) - \nabla I_\mathrm{c}(p)} \leq \zeta_Q H\divx{p}{q}.
    \end{align*}
    The desired result then follows from adding the above inequalities with a copy of themselves with the roles of $p$ and $q$ reversed, then comparing to Proposition~\ref{prop:relative-alt}(a-ii) and~\ref{prop:relative-alt}(b-ii). Parts (b) and (c) can be established using a similar proof. Alternatively, we can use the results from~\cite{ramakrishnan2020computing} combined with Proposition~\ref{cor:mono-rel}.
\end{proof}
\begin{rem}\label{rem:coeff-bound}
    Due to monotonicity~\cite{ruskai2002inequalities} and non-negativity of classical and quantum relative entropy, the channel coefficients always satisfy the bounds $0\leq\eta_{(\wc)}\leq\zeta_{(\wc)}\leq1$. Therefore, negative classical mutual information and Holevo information are always at least $1$-smooth relative to negative Shannon entropy, and negative quantum mutual information is always at least $2$-smooth relative to negative von Neumann entropy. In general, it is not straightforward to efficiently estimate or obtain tighter bounds on these channel coefficients, as they are given by non-convex optimization problems. In~\cite{faust2022strong}, it was shown how upper bounds on the contraction coefficient of classical channels could be found by using approximation and sum-of-squares techniques. However, tighter bounds on these channel coefficients are not needed for certain variations on PDHG, such as the backtracking PDHG algorithm presented in Appendix~\ref{appdx:bcktrack-pdhg}, to take advantage of well-conditioned problem instances while still enjoying convergence guarantees.
\end{rem}

Using these relative smoothness properties combined with strong convexity of Shannon entropy and von Neumann entropy from Proposition~\ref{prop:entr-strong} to ensure the existence of suitable step sizes, we can therefore apply PDHG (with or without backtracking) to solve for each energy-constrained channel capacity while achieving the ergodic sublinear convergence guarantees provided by Proposition~\ref{prop:pdhg-conv}. To implement PDHG for the energy-constrained classical and cq channel capacity, we use the kernel function $\varphi(\wc)=-H(\wc)$, the primal domain $\mathcal{C}=\Delta_m$, and all energy-constraints are dualized. This results in the following primal-dual iterates $(p^k, \lambda^k)$
\begin{subequations}
    \begin{align}
        & \bar{\lambda}^{k+1} = \lambda^k + \theta_k (\lambda^k - \lambda^{k-1})\\
        & p^{k+1}_j = \frac{\bar{p}_j^{k+1}}{\sum_{j=1}^m \bar{p}_j^{k+1}},\quad \textrm{ for } j=1,\ldots,m\label{eqn:pdhg-class-cc-primal}\\
        & \lambda^{k+1} = (\lambda^k + \gamma_k (Ap^{k+1} - b))_+, \label{eqn:pdhg-class-cc-dual}
    \end{align}
\end{subequations}
where 
\begin{equation}
    \bar{p}_j^{k+1} = p^k_j \exp(\tau_k (\partial_j f(p^k) - \inp{A_j}{\bar{\lambda}^{k+1}}),
\end{equation}
$\lambda\in\mathbb{R}^l_+$ is the dual variable corresponding to the energy constraints, $\partial_j f$ is the partial derivative of $f$ with respect to the $j$-th coordinate, $A_j$ is the $j$-th column of $A$, and we define either $f\coloneqq I_\mathrm{c}$ for the classical channel capacity or $f\coloneqq\chi$ for the cq channel capacity. 

Similarly, for the energy-constrained ea channel capacity, we use the kernel function $\varphi=-S$, the primal domain $\mathcal{C}=\mathcal{D}(\mathcal{H}_A)$, and all energy constraints are dualized. The resulting PDHG iterates $(\rho^k, \lambda^k)$ are of the form
\begin{subequations}
    \begin{align}
        & \bar{\lambda}^{k+1} = \lambda^k + \theta_k (\lambda^k - \lambda^{k-1})\\
        & \rho^{k+1} = \frac{\bar{\rho}^{k+1}}{\tr[\bar{\rho}^{k+1}]} \label{eqn:pdhg-class-ea-primal}\\
        & \lambda^{k+1}_i = (\lambda^k_i + \gamma_k (\tr[A_i\rho] - b_i))_+, \quad \textrm{ for } i=1\ldots,l,
    \end{align}
\end{subequations}
where 
\begin{equation}
    \bar{\rho}^{k+1} = \exp\biggl(\log(\rho^k) + \tau_k \biggl(\nabla I_\mathrm{q}(\rho^k) - \sum_{i=1}^l \bar{\lambda}_i^{k+1}A_i \biggr) \biggr) ,
\end{equation}
$\lambda\in\mathbb{R}^l_+$ is the dual variable corresponding to the energy constraints. We remark that the primal iterates~\eqref{eqn:pdhg-class-cc-primal} are identical to the Blahut-Arimoto methods for unconstrained channel capacities~\cite{blahut1972computation,li2019blahut}, and PDHG essentially augments these existing Blahut-Arimoto methods for constrained channel capacities with a single additional dual ascent step which adaptively solves for the Lagrange multipliers for specific energy bounds $b$.

\begin{figure*}[]
\centering
\pgfplotstableread[col sep=comma]{figures/data/conv.csv}\data
\begin{tikzpicture}
    \definecolor{clr1}{RGB}{0,0,0}
    \definecolor{clr2}{RGB}{33,33,33}
    \definecolor{clr3}{RGB}{66,66,66}
    \definecolor{clr4}{RGB}{100,100,100}
    \definecolor{clr5}{RGB}{133,133,133}
    
    \begin{semilogyaxis}[
        name=mainplot,
        width = 0.375\textwidth,
        height = 0.375\textwidth,
        xmin = 0.0, xmax = 456,
        ymin = 1e-7, ymax=1,
        xlabel=Iteration $k$,
        ylabel=Optimality gap (nats),
        label style = {font=\footnotesize},
        ticklabel style = {font=\footnotesize},
        title style = {font=\footnotesize},
        legend pos  = north east,
        legend cell align={left},
        legend style={fill=white, fill opacity=0.6, draw opacity=1,text opacity=1, font=\footnotesize},
        no markers,
        every axis plot/.append style={very thick},
        grid=major,
        major grid style={line width=.1pt,draw=gray!20},
        yminorticks=false,
        title = {(a) Classical channel capacity}
        ]
        
        \addplot[clr5, dashdotted] table[x=i , y=cc4] {\data}; 
        \addplot[clr3, densely dashed] table[x=i , y=cc128] {\data};
        \addplot[clr1] table[x=i , y=cc8192] {\data};

        \legend{
            {$n=4, l=1$},
            {$n=128, l=4$},
            {$n=8192, l=8$}
        }

    \end{semilogyaxis}
\end{tikzpicture}
\begin{tikzpicture}
    \definecolor{clr1}{RGB}{0,0,0}
    \definecolor{clr2}{RGB}{33,33,33}
    \definecolor{clr3}{RGB}{66,66,66}
    \definecolor{clr4}{RGB}{100,100,100}
    \definecolor{clr5}{RGB}{133,133,133}

    \begin{semilogyaxis}[
        name=secondplot,
        at={(mainplot.north east)},
        xshift=0.05cm,
        anchor=north west,   
        width = 0.375\textwidth,
        height = 0.375\textwidth,
        xmin = 0.0, xmax = 261,
        ymin = 1e-7, ymax=1,
        xlabel=Iteration $k$,
        label style = {font=\footnotesize},
        ticklabel style = {font=\footnotesize},
        title style = {font=\footnotesize},
        legend pos  = north east,
        legend cell align={left},
        legend style={fill=white, fill opacity=0.6, draw opacity=1,text opacity=1, font=\footnotesize},
        no markers,
        every axis plot/.append style={very thick},
        grid=major,
        major grid style={line width=.1pt,draw=gray!20},
        yminorticks=false,
        yticklabels={,,},
        title = {(b) cq channel capacity}
        ]
        
        \addplot[clr5, dashdotted] table[x=i , y=cq4] {\data}; 
        \addplot[clr3, densely dashed] table[x=i , y=cq32] {\data};
        \addplot[clr1] table[x=i , y=cq128] {\data};

        \legend{
            {$n=4, l=1$},
            {$n=32, l=4$},
            {$n=128, l=6$}
        }
        
    \end{semilogyaxis}
\end{tikzpicture}
\begin{tikzpicture}
    \definecolor{clr1}{RGB}{0,0,0}
    \definecolor{clr2}{RGB}{33,33,33}
    \definecolor{clr3}{RGB}{66,66,66}
    \definecolor{clr4}{RGB}{100,100,100}
    \definecolor{clr5}{RGB}{133,133,133}

    \begin{semilogyaxis}[
        at={(secondplot.north east)},
        xshift=0.28cm,
        anchor=north west,  
        width = 0.375\textwidth,
        height = 0.375\textwidth,
        xmin = 0.0, xmax = 199,
        ymin = 1e-7, ymax=1,
        xlabel=Iteration $k$,
        label style = {font=\footnotesize},
        ticklabel style = {font=\footnotesize},
        title style = {font=\footnotesize},
        legend pos  = north east,
        legend cell align={left},
        legend style={fill=white, fill opacity=0.6, draw opacity=1,text opacity=1, font=\footnotesize},
        no markers,
        every axis plot/.append style={very thick},
        grid=major,
        major grid style={line width=.1pt,draw=gray!20},
        yminorticks=false,
        yticklabels={,,},
        title = {(c) ea channel capacity}
        ]
        
        \addplot[clr5, dashdotted] table[x=i , y=ea4] {\data}; 
        \addplot[clr3, densely dashed] table[x=i , y=ea16] {\data};
        \addplot[clr1] table[x=i , y=ea64] {\data};

        \legend{
            {$n=4, l=1$},
            {$n=16, l=3$},
            {$n=64, l=5$}
        }
        
    \end{semilogyaxis}
\end{tikzpicture}
% \vspace{-1.5em}
\caption{Convergence of PDHG to solve (a) the classical channel capacity with $n$ inputs and outputs and $l$ energy constraints, (b) the classical-quantum channel capacity with an alphabet size and channel input/output dimension of $n$ and $l$ energy constraints, and (c) the entanglement-assisted channel capacity with channel input/output dimension of $n$ and $l$ energy constraints.}
\label{fig:channel-capacities}
\end{figure*}
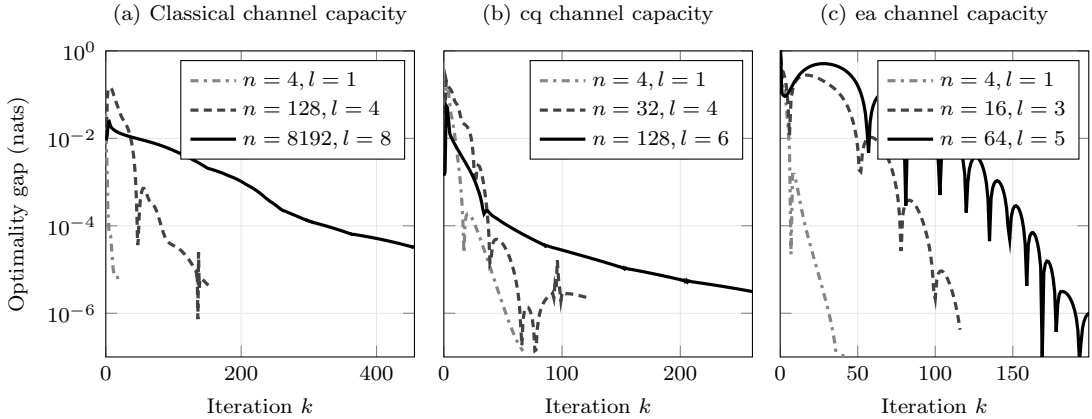

\begin{table}[t]
\caption{Computation of classical channel capacity with $n$ inputs and outputs and $l$ energy constraints.}\label{tab:cc}
\vspace{-0.5em}
\footnotesize
\centering
\begin{tabular*}{\columnwidth}{@{\extracolsep{\fill}}ccccccc@{\extracolsep{\fill}}}
\toprule
\multicolumn{1}{c}{} &  & \multicolumn{3}{c}{PDHG ($\kappa=1$)} & \multicolumn{2}{c}{MOSEK w/ CVX} \\ \cmidrule{3-5}\cmidrule{6-7}
$n$ & $l$ & Time (s) & Iter. & Opt. gap (nats) & Time (s) & Opt. gap (nats) \\ \midrule
$4$ & $1$ & $0.009$ & $19$ & $4.9\times10^{-6}$ & $0.990$ & $1.5\times10^{-8}$ \\
$128$ & $4$ & $0.044$ & $153$ & $4.2\times10^{-6}$ & $1.039$ & $1.5\times10^{-8}$ \\
$8192$ & $8$ & $48.12$ & $456$ & $3.2\times10^{-5}$ & $213.6$ & $1.5\times10^{-8}$ \\ \bottomrule
\end{tabular*}

\bigskip
\caption{Computation of classical-quantum channel capacity with an alphabet size and channel input/output dimension of $n$ and $l$ energy constraints.}\label{tab:cq}
\vspace{-0.5em}
\footnotesize
\centering
\begin{tabular*}{\columnwidth}{@{\extracolsep{\fill}}ccccccc@{\extracolsep{\fill}}}
\toprule
\multicolumn{1}{c}{} &  & \multicolumn{3}{c}{PDHG ($\kappa=1$)} & \multicolumn{2}{c}{MOSEK w/ CVXQUAD} \\ \cmidrule{3-5}\cmidrule{6-7}
$n$ & $l$ & Time (s) & Iter. & Opt. gap (nats) & Time (s) & Opt. gap (nats) \\ \midrule
$4$ & $1$ & $0.016$ & $67$ & $1.4\times10^{-7}$ & $1.462$ & $1.5\times10^{-8}$ \\
$32$ & $4$ & $0.176$ & $122$ & $2.3\times10^{-6}$ & $4.338$ & $7.0\times10^{-10}$ \\
$128$ & $6$ & $3.194$ & $261$ & $3.1\times10^{-6}$ & $1931.$ & $1.5\times10^{-8}$ \\ \bottomrule
\end{tabular*}

\bigskip
\caption{Computation of entanglement-assisted channel capacity with channel input/output dimension of $n$ and $l$ energy constraints.}\label{tab:ea}
\vspace{-0.5em}
\footnotesize
\centering
\begin{tabular*}{\columnwidth}{@{\extracolsep{\fill}}ccccccc@{\extracolsep{\fill}}}
\toprule
\multicolumn{1}{c}{} &  & \multicolumn{3}{c}{PDHG ($\kappa=10$)} & \multicolumn{2}{c}{MOSEK w/ CVXQUAD} \\ \cmidrule{3-5}\cmidrule{6-7}
$n$ & $l$ & Time (s) & Iter. & Opt. gap (nats) & Time (s) & Opt. gap (nats) \\ \midrule
$4$ & $1$ & $0.113$ & $46$ & $1.4\times10^{-7}$ & $27.28$ & $2.5\times10^{-9}$ \\
$16$ & $3$ & $1.193$ & $116$ & $4.2\times10^{-7}$ & \multicolumn{2}{c}{Not enough memory} \\
$64$ & $5$ & $41.09$ & $200$ & $1.0\times10^{-6}$ & \multicolumn{2}{c}{Not enough memory} \\ \bottomrule
\end{tabular*}
\end{table}

We show experimental results for the computation of energy-constrained variants of the classical channel capacity in Table~\ref{tab:cc} and Figure~\ref{fig:channel-capacities}(a), the cq channel capacity in Table~\ref{tab:cq} and Figure~\ref{fig:channel-capacities}(b), and the ea channel capacity in Table~\ref{tab:ea} and Figure~\ref{fig:channel-capacities}(c). All classical and quantum channels were randomly generated. For classical and cq channel capacities, the energy constraint matrix $A$ was randomly generated using a uniform random distribution between $0$ and $1$. For the ea channel capacity, observables $A_i$ are generated by scaling random density matrices so that their trace is equal to the dimension of the channel. Energy bounds $b$ are uniformly random generated between $0$ and $1$. If an infeasible set of energy constraints was sampled or if none of the constraints were active at the optimal solution, we rejected it and resampled the energy constraints. Results show that PDHG is able to solve to medium-accuracy solutions up to $500$ times faster than MOSEK. Additionally, PDHG is able to successfully solve large-scale problem instances that the off-the-shelf solvers were unable to solve due to lack of memory. In some of the plots in Figure~\ref{fig:channel-capacities}, it appears that PDHG converges linearly to the solution.

\subsection{Rate-Distortion Functions} \label{sec:rd}

Consider the classical channel setup that was defined in Section~\ref{subsec:cc-pd} for the classical channel capacity problem~\eqref{eqn:cc}. Let us define a distortion matrix $\delta\in\mathbb{R}^{n\times m}_+$ where $\delta_{ij}$ represents the distortion of producing output $i$ from input $j$, and $P\in\Delta_{n\times m}$ be the joint distribution such that $P_{ij}$ is the probability of having output $i$ and input $j$. The \emph{classical rate-distortion function} $R_c(D)$ for an input probability distribution $p\in\Delta_m$ quantifies how much a signal can be compressed by a lossy channel while remaining under a maximum allowable distortion $D\geq0$ of the signal, and is given by~\cite{shannon1948mathematical}
\begin{subequations}\label{eqn:crd}
    \begin{align}
        \minimize_{P\in\Delta_{n\times m}} \quad & I_\mathrm{c}(P) \\
        \mathmakebox[\widthof{$\minimize_{P\in\Delta_{n\times m}}$}]{\subjto}\, & \sum_{i=1}^m P_{ij} = p_j, \qquad \textrm{ for } j=1,\ldots,m,\\
                \quad & \sum_{j=1}^m \sum_{i=1}^n P_{ij}\delta_{ij} \leq D, \label{eqn:crd-c}
    \end{align}
\end{subequations}
where, with slight abuse of notation, $I_\mathrm{c}:\mathbb{R}^{n\times m}_+\rightarrow\mathbb{R}$ is the same classical mutual information function introduced in~\eqref{eqn:cc}, but we now treat the input distribution $p$ as fixed problem data, and the joint distribution $P$ as the variable, i.e.,
\begin{equation}
    I_\mathrm{c}(P) \coloneqq \sum_{j=1}^m \sum_{i=1}^n P_{ij} \log\mleft(\frac{P_{ij}}{p_j \sum_{k=1}^m P_{ik}}\mright).
\end{equation}
Note that~\eqref{eqn:crd} is expressed differently but is equivalent to the original problem defined in~\cite{blahut1972computation}. The partial derivatives of classical mutual information with respect to the joint distribution is
\begin{equation}
    \frac{\partial I_\mathrm{c}}{\partial P_{ij}} = \log(P_{ij}) - \log\biggl( p_j \sum_{k=1}^m P_{ik} \biggr).
\end{equation}

Now consider the ea quantum channel setup that was defined in Section~\ref{subsec:cc-pd} for the ea channel capacity problem. Given a rank-$n$ input state $\rho_A=\sum_{i=1}^n \lambda_i \ket{a_i}\!\bra{a_i}$ where $\lambda_i > 0$ for all $i=1,\ldots,n$, define the purified state $\ket{\psi} \in \mathcal{H}_A \otimes \mathcal{H}_R$ as
\begin{equation}
    \ket{\psi} = \sum_{i=1}^n \sqrt{\lambda_i} \ket{a_i} \otimes \ket{r_i}.
\end{equation}
Here, $\mathcal{H}_R$ is a reference system, $\dim\mathcal{H}_R\geq n$, and $\{ \ket{r_i} \}$ is some orthonormal basis on $\mathcal{H}_R$. Let $\Delta \in \mathcal{B}(\mathcal{H}_B\otimes\mathcal{H}_R)_+$ be a positive semidefinite distortion observable. The \emph{ea rate distortion function} $R_q(D)$ for maximum allowable distortion $D\geq0$ and input state $\rho_A\in\mathcal{D}(\mathcal{H}_A)$ is~\cite{datta2012quantum,wilde2013quantum}
\begin{subequations}\label{eqn:qrd}
    \begin{align}
        \minimize_{\rho_{BR} \in \mathcal{D}(\mathcal{H}_B\otimes\mathcal{H}_R)} \quad & I_\mathrm{q}(\rho_{BR}) \\
        \mathmakebox[\widthof{$\minimize_{\rho_{BR} \in \mathcal{D}(\mathcal{H}_B\otimes\mathcal{H}_R)}$}][c]{\subjto}\, & \tr_B(\rho_{BR}) = \rho_R, \label{eqn:qrd-b}\\
                            \quad & \inp{\Delta}{\rho_{BR}} \leq D, \label{eqn:qrd-c}
    \end{align}
\end{subequations}
where $\rho_R\coloneqq\tr_A(\ket{\psi}\!\bra{\psi})$. Again, with a slight abuse of notation, $I_\mathrm{q}:\mathcal{B}(\mathcal{H}_B\otimes\mathcal{H}_R)_+\rightarrow\mathbb{R}$ is the quantum mutual information function introduced in~\eqref{eqn:ea}, but we now treat the input state $\rho_A$ as fixed problem data, and the bipartite output state $\rho_{BR}$ as a variable, i.e.,
\begin{equation}
    I_\mathrm{q}(\rho_{BR}) \coloneqq S(\rho_A) + S(\tr_R(\rho_{BR})) - S(\rho_{BR}).
\end{equation}
Using Corollaries~\ref{cor:grad-tr-f} and~\ref{cor:grad-tr-f-lin} (in Appendix~\ref{appdx:grad}), and recognizing that the adjoint of the partial trace operator is $(\tr_1)^\dag(X_2) = \mathbb{I}_1 \otimes X_2$, the gradient of quantum mutual information with respect to $\rho_{BR}$ is
\begin{equation}\label{eqn:qmi-qrd}
    \nabla I_\mathrm{q}(\rho_{BR}) = \log(\rho_{BR}) - \log(\tr_R(\rho_{BR})) \otimes \mathbb{I}_R.
\end{equation}

Note that classical and quantum mutual information, respectively, can be expressed as
\begin{gather}
    I_\mathrm{c}(P) = H(p) - H\divy{X}{Y}_{P},\\
    I_\mathrm{q}(\rho_{BR}) = S(\rho_R) - S\divy{R}{B}_{\rho_{BR}}.
\end{gather}
As classical and quantum conditional entropy are concave in $P$ and $\rho_{BR}$, respectively, both classical and quantum mutual information are convex functions in $P$ and $\rho_{BR}$, respectively. Given that all constraints are linear and that $\Delta_{n,m}$ and $\mathcal{D}(\mathcal{H}_B\otimes\mathcal{H}_R)$ are convex sets, both~\eqref{eqn:crd} and~\eqref{eqn:qrd} are convex optimization problems.

We now show that the rate-distortion problems share similar relative smoothness properties to their channel capacity counterparts.
\begin{thm}\label{thm:qrd-smooth}
    Consider any input distribution $p\in\Delta_m$ and input state $\rho_A\in\mathcal{D}(\mathcal{H}_A)$. The following relative smoothness properties hold:
    \begin{enumerate}[label=(\alph*)]
        \item Classical mutual information $I_\mathrm{c}$ is $1$-smooth relative to negative Shannon entropy $-H(\wc)$ on $\Delta_{n\times m}$.
        \item Quantum mutual information $I_\mathrm{q}$ is $1$-smooth relative to negative von Neumann entropy $-S(\wc)$ on $\mathcal{D}(\mathcal{H}_B\otimes\mathcal{H}_R)$.
    \end{enumerate}    
\end{thm}
\begin{proof}
    We will only show the proof for (b), as (a) can be established using essentially the same method. Using the fact that partial trace and tensor product operators are adjoint to each other, we can show
    \begin{align*}
        \inp{\nabla I_\mathrm{q}(\rho_{BR}) - \nabla I_\mathrm{q}(\sigma_{BR})}{\rho_{BR}} &= S\divx{\rho_{BR}}{\sigma_{BR}} - S\divx{\tr_R(\rho_{BR})}{\tr_R(\sigma_{BR})}\\
        &\leq S\divx{\rho_{BR}}{\sigma_{BR}}.
    \end{align*}  
    Using Remark~\ref{rem:mono} gives the desired result.
\end{proof}

We can therefore use PDHG to solve for the rate-distortion functions while guaranteeing ergodic sublinear convergence from Proposition~\ref{prop:pdhg-conv}, where strong convexity of Shannon and von Neumann entropy guarantees the existence of suitable step sizes to achieve this rate. To implement PDHG for the classical rate distortion function, we use the kernel function $\varphi(\wc)=-H(\wc)$, the primal domain $\mathcal{C}=\{ P\in\Delta_{m,n} : \sum_{i=1}^m P_{ij} = p_j \}$, and dualize the distortion constraint~\eqref{eqn:crd-c}. This results in the primal-dual iterates $(P^k, \lambda^k)$ generated by
\begin{subequations}
    \begin{align}
        & \bar{\lambda}^{k+1} = \lambda^k + \theta_k (\lambda^k - \lambda^{k-1})\\
        & P^{k+1}_{ij} = \frac{p_j\bar{P}_{ij}^{k+1}}{\sum_{i=1}^m \bar{P}_{ij}^{k+1}}, \qquad \textrm{ for } \begin{array}{l}
        i=1,\ldots,n,\\
        j=1,\ldots,m
        \end{array} \label{eqn:pdhg-class-rd-primal}\\
        & \lambda^{k+1} = \biggl(\lambda^k + \gamma_k \biggl(\sum_{j=1}^m \sum_{i=1}^n P_{ij}\delta_{ij} - D \biggr)\biggr)_+,
    \end{align}
\end{subequations}
where 
\begin{equation}
    \bar{P}_{ij}^{k+1} = P^k_{ij} \exp(-\tau_k (\partial_{ij} I_\mathrm{c}(P^k) + \lambda\delta_{ij}),
\end{equation}
$\lambda\in\mathbb{R}_+$ is the dual variable corresponding to the distortion constraint and $\partial_{ij} I_\mathrm{c}$ is the partial derivative of $I_\mathrm{c}$ with respect to the $(i,j)$-th coordinate. If $\tau_k=1$, then the primal iterate~\eqref{eqn:pdhg-class-rd-primal} is identical to the original step proposed by Blahut~\cite{blahut1972computation} up to a change of variables. This implies that mirror descent applied to the Lagrangian for a fixed Lagrange multiplier $\lambda$ recovers the  original Blahut-Arimoto algorithm for classical rate-distortion functions. Moreover, like for the energy-constrained channel capacities, PDHG has a natural interpretation of including an additional dual ascent step to adaptively find the Lagrange multiplier which solves for a given distortion $D$. 

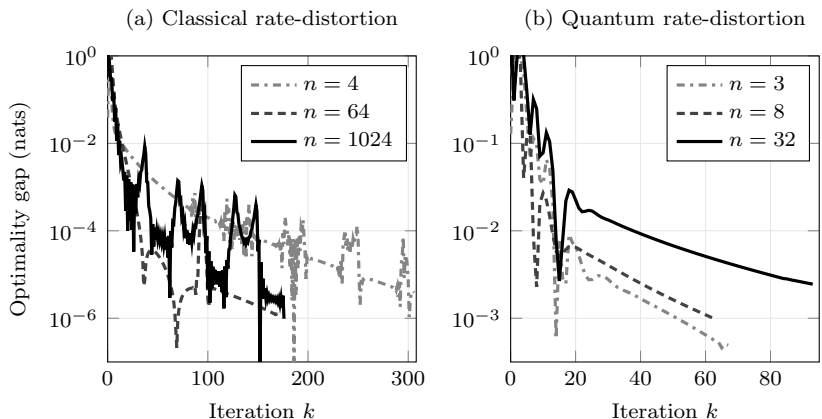
\begin{figure*}[]
\centering
\pgfplotstableread[col sep=comma]{figures/data/conv.csv}\data
\begin{tikzpicture}
    \definecolor{clr1}{RGB}{0,0,0}
    \definecolor{clr2}{RGB}{33,33,33}
    \definecolor{clr3}{RGB}{66,66,66}
    \definecolor{clr4}{RGB}{100,100,100}
    \definecolor{clr5}{RGB}{133,133,133}
    
    \begin{semilogyaxis}[
        name=mainplot,
        width = 0.375\textwidth,
        height = 0.375\textwidth,
        xmin = 0.0, xmax = 308,
        ymin = 1e-7, ymax=1,
        xlabel=Iteration $k$,
        ylabel=Optimality gap (nats),
        label style = {font=\footnotesize},
        ticklabel style = {font=\footnotesize},
        title style = {font=\footnotesize},
        legend pos  = north east,
        legend cell align={left},
        legend style={fill=white, fill opacity=0.6, draw opacity=1,text opacity=1, font=\footnotesize},
        no markers,
        every axis plot/.append style={very thick},
        grid=major,
        major grid style={line width=.1pt,draw=gray!20},
        yminorticks=false,
        title = {(a) Classical rate-distortion}
        ]
        
        \addplot[clr5, dashdotted] table[x=i , y=crd4] {\data}; 
        \addplot[clr3, densely dashed] table[x=i , y=crd64] {\data};
        \addplot[clr1] table[x=i , y=crd1024] {\data};

        \legend{
            {$n=4$},
            {$n=64$},
            {$n=1024$}
        }

    \end{semilogyaxis}
\end{tikzpicture}
\begin{tikzpicture}
    \definecolor{clr1}{RGB}{0,0,0}
    \definecolor{clr2}{RGB}{33,33,33}
    \definecolor{clr3}{RGB}{66,66,66}
    \definecolor{clr4}{RGB}{100,100,100}
    \definecolor{clr5}{RGB}{133,133,133}

    \begin{semilogyaxis}[
        name=secondplot,
        at={(mainplot.north east)},
        xshift=0.05cm,
        anchor=north west,   
        width = 0.375\textwidth,
        height = 0.375\textwidth,
        xmin = 0.0, xmax = 95,
        ymin = 3.16e-4, ymax=1,
        xlabel=Iteration $k$,
        label style = {font=\footnotesize},
        ticklabel style = {font=\footnotesize},
        title style = {font=\footnotesize},
        legend pos  = north east,
        legend cell align={left},
        legend style={fill=white, fill opacity=0.6, draw opacity=1,text opacity=1, font=\footnotesize},
        no markers,
        every axis plot/.append style={very thick},
        grid=major,
        major grid style={line width=.1pt,draw=gray!20},
        yminorticks=false,
        title = {(b) Quantum rate-distortion}
        ]
        
        \addplot[clr5, dashdotted] table[x=i , y=qrd3] {\data}; 
        \addplot[clr3, densely dashed] table[x=i , y=qrd8] {\data};
        \addplot[clr1] table[x=i , y=qrd32] {\data};

        \legend{
            {$n=3$},
            {$n=8$},
            {$n=32$}
        }
        
    \end{semilogyaxis}
\end{tikzpicture}
% \vspace{-1.5em}
\caption{Convergence of PDHG to solve (a) the classical rate-distortion for a Hamming distortion measure with $D=0.5$ and a channel with $n$ inputs and outputs, and (b) the quantum rate-distortion for an entanglement fidelity distortion measure with $D=0.5$ and a channel with input/output dimension of $n$.}
\label{fig:rate-distortions}
\end{figure*}

\begin{table}[t]

\caption{Computation of the classical rate-distortion for a Hamming distortion measure with $D=0.5$ and a channel with $n$ inputs and outputs.}\label{tab:crd}
\vspace{-0.5em}
\footnotesize
\centering
\begin{tabular*}{\columnwidth}{@{\extracolsep{\fill}}cccccc@{\extracolsep{\fill}}}
\toprule
\multicolumn{1}{c}{} & \multicolumn{3}{c}{PDHG ($\kappa=10$)} & \multicolumn{2}{c}{MOSEK w/ CVX} \\ \cmidrule{2-4}\cmidrule{5-6}
$n$ & Time (s) & Iter. & Opt. gap (nats) & Time (s) & Opt. gap (nats) \\ \midrule
$4$ & $0.018$ & $306$ & $3.9\times10^{-6}$ & $0.951$ & $1.5\times10^{-8}$ \\
$64$ & $0.094$ & $170$ & $1.4\times10^{-6}$ & $1.900$ & $1.5\times10^{-8}$ \\
$1024$ & $35.57$ & $176$ & $2.6\times10^{-6}$ & \multicolumn{2}{c}{Not enough memory} \\ \bottomrule
\end{tabular*}

\bigskip
\caption{Computation of the quantum rate-distortion for an entanglement fidelity distortion measure with $D=0.5$ and a channel with input/output dimension of $n$.}\label{tab:qrd}
\vspace{-0.5em}
\footnotesize
\centering
\begin{tabular*}{\columnwidth}{@{\extracolsep{\fill}}cccccc@{\extracolsep{\fill}}}
\toprule
\multicolumn{1}{c}{} & \multicolumn{3}{c}{PDHG ($\kappa=10$)} & \multicolumn{2}{c}{MOSEK w/ CVXQUAD} \\ \cmidrule{2-4}\cmidrule{5-6}
$n$ & Time (s) & Iter. & Opt. gap (nats) & Time (s) & Opt. gap (nats) \\ \midrule
$3$ & $0.094$ & $68$ & $4.9\times10^{-4}$ & $83.40$ & $5.2\times10^{-6}$ \\
$8$ & $0.261$ & $62$ & $1.0\times10^{-3}$ & \multicolumn{2}{c}{Not enough memory} \\
$32$ & $37.95$ & $93$ & $2.5\times10^{-3}$ & \multicolumn{2}{c}{Not enough memory} \\ \bottomrule
\end{tabular*}
\end{table}

To implement PDHG for the quantum rate-distortion function, we use the kernel function $\varphi(\wc)=-S(\wc)$, the primal domain $\mathcal{C}=\mathcal{D}(\mathcal{H}_B\otimes\mathcal{H}_R)$, and dualize both the partial trace~\eqref{eqn:qrd-b} and distortion constraint~\eqref{eqn:qrd-c}. This gives us the primal-dual iterates $(\rho_{BR}^k, \nu^k, \lambda^k)$ generated by
\begin{subequations}
    \begin{align}
        & \bar{\nu}^{k+1} = \nu^k + \theta_k (\nu^k - \nu^{k-1})\\
        & \bar{\lambda}^{k+1} = \lambda^k + \theta_k (\lambda^k - \lambda^{k-1})\\
        & \rho^{k+1}_{BR} = \frac{\bar{\rho}_{BR}^{k+1}}{\tr[\bar{\rho}_{BR}^{k+1}]} \\
        & \nu^{k+1} = \nu^k + \gamma_k (\tr_B(\rho_{BR}^{k+1}) - \rho_R)\\
        & \lambda^{k+1} = (\lambda^k + \gamma_k (\inp{\Delta}{\rho_{BR}^{k+1}} - D))_+,
    \end{align}
\end{subequations}
where 
\begin{equation}
    \bar{\rho}_{BR}^{k+1} = \exp(\log(\rho^k_{BR}) - \tau_k(\nabla I_\mathrm{q}(\rho^k_{BR}) + \mathbb{I}\otimes\bar{\nu}^{k+1} + \bar{\lambda}^{k+1}\Delta)),
\end{equation}
$\nu\in\mathcal{B}(\mathcal{H}_R)$ is the dual variable corresponding to the partial trace constraint, and $\lambda\in\mathbb{R}_+$ is the dual variable corresponding to the distortion constraint. 

We show experimental results for the computation of the classical rate-distortion with Hamming distortion $\delta=\bm{1}\bm{1}^\top-\mathbb{I}$ in Table~\ref{tab:crd} and Figure~\ref{fig:rate-distortions}(a), and the quantum rate-distortion with entanglement fidelity distortion $\Delta=\mathbb{I}-\ket{\psi}\!\bra{\psi}$ in Table~\ref{tab:qrd} and Figure~\ref{fig:rate-distortions}(b). All experiments use $D=0.5$, and all input states are randomly generated. Similar to the channel capacity experiments, the results show that PDHG solves up to $1000$ times faster than MOSEK to low- to medium-accuracy solutions, and are able to effectively scale to large scale problem dimensions. We also observe a combination of sublinear and linear convergence rates of PDHG in Figure~\ref{fig:rate-distortions}.

\subsection{Relative Entropy of a Quantum Resource}\label{subsec:ree}

Quantum resource theories~\cite{brandao2015reversible,chitambar2019quantum} aim to categorize quantum resources that can be generated using a set of permissible physical operations. The set of states which can be generated are called the free states $\mathcal{S}(\mathcal{H})$ associated with a Hilbert space. Often, we are interested in quantifying the relative entropy of a quantum resource possessed by a state $\rho\in\mathcal{D}(\mathcal{H})$, which is defined as
\begin{equation}\label{eqn:rel-entr-res}
    \min_{\sigma\in\mathcal{S}(\mathcal{H})} S\divx{\rho}{\sigma}.
\end{equation}
Note that as $\rho$ is fixed problem data, this is equivalent to
\begin{equation}
    S(\rho) + \min_{\sigma\in\mathcal{S}(\mathcal{H})} g(\sigma),
\end{equation}
where $g(\sigma) \coloneqq -\tr[\rho\log(\sigma)]$. To find the gradient of quantum relative entropy with respect to the second argument, first note that due to linearity of the trace operator, the directional derivative along $V\in\mathbb{H}^n$ can be found using Lemma~\ref{lem:dir-deriv}(a) (in Appendix~\ref{appdx:grad}) as
\begin{align*}
    \mathsf{D}g(\sigma)[V] &= -\inp{\rho}{U [f^{[1]}(\Lambda) \odot (U^\dag V U)] U^\dag}\\
    &= -\inp{V}{U [f^{[1]}(\Lambda) \odot (U^\dag \rho U)] U^\dag},
\end{align*}
where $\sigma$ has diagonalization $\sigma=U\Lambda U^\dag$ and $f^{[1]}(\Lambda)$ is the first divided difference matrix (see~\eqref{eqn:fdd} in Appendix~\ref{appdx:grad}) associated with $f(x)=\log(x)$. For the second equality, we recognize that $\inp{X}{Y\odot Z} =\inp{Y\odot X}{Z}$ for $X,Z\in\mathbb{H}^n$ and $Y\in\mathbb{R}^{n\times n}$. It follows that the gradient of quantum relative entropy with respect to $\sigma$ is
\begin{align}
    \grad_\sigma S\divx{\rho}{\sigma} = \nabla g(\sigma) = -U [f^{[1]}(\Lambda) \odot (U^\dag \rho U)] U^\dag.
\end{align}
As quantum relative entropy is jointly convex, if the set of free states $\mathcal{S}(\mathcal{H})$ is convex, then the relative entropy of a quantum resource is computed by solving a convex optimization problem.

Before deriving a suitable algorithm to compute the relative entropy of a quantum resource, we will first show that the von Neumann entropy is not a suitable kernel to compute this quantity.
\begin{prop}
    Quantum relative entropy $S\divx{\rho}{\wc}$ with fixed first argument $\rho\in\mathcal{D}(\mathcal{H})$ is not smooth relative to von Neumann entropy.
\end{prop}
\begin{proof}
    It suffices to study the univariate scenario where $f(x)=-a\log(x)$, $\varphi(x)=x\log(x)$, and $a,x\geq0$. The rest of the proof follows from~\cite[Appendix A]{li2019convergence}.
\end{proof}

Therefore, to apply mirror descent or PDHG, we need to find a different suitable kernel function. To do this, we will present a generalized result which can be applied to problems with similar objectives.
\begin{thm}\label{thm:spec-mono}
    Consider a function $f:\mathbb{R}_{++}\rightarrow\mathbb{R}$ which is operator convex, i.e., for all $n\in\mathbb{N}$, $X,Y\in\mathbb{H}^n_{++}$ and $\lambda\in[0,1]$,
    \begin{equation*}
        f(\lambda X + (1-\lambda)Y) \preceq \lambda f(X) + (1 - \lambda) f(Y).
    \end{equation*}
    Then for any $A\succeq0$, the function $g(X) = \tr[Af(X)]$ is $\lambda_{\mathrm{max}}(A)$-smooth and $\lambda_{\mathrm{min}}(A)$-strongly convex relative to $\tr[f(\wc)]$ on $\mathbb{H}^n_{++}$.
\end{thm}
\begin{proof}
    We first recognize that as the trace of the product of two positive semidefinite matrices is always non-negative, it follows from operator convexity of $f$ that $g$ is convex for all $A\succeq0$. Therefore, as $\lambda_{\mathrm{max}}(A)\mathbb{I}-A\succeq0$, it follows that
    \begin{equation*}
        \tr[(\lambda_{\mathrm{max}}(A)\mathbb{I}-A)f(X)] = \lambda_{\mathrm{max}}(A)\tr[f(X)] - \tr[Af(X)],
    \end{equation*}
    must also be convex, which by definition recovers the result for relative smoothness. Similarly, noting that $A-\lambda_{\mathrm{min}}(A)\mathbb{I}\succeq0$ and applying a similar argument recovers the relative strong convexity result.
\end{proof}

\begin{rem}\label{rem:trace-func-kernel}
    Trace functions are convenient to use as kernel functions, as finding an expression for the mirror descent iterates~\eqref{eqn:mirror-descent} reduces to solving a set of univariate equations. Specifically, assuming $\varphi(X)=\tr[h(X)]$ for some function $h:\mathbb{R}\rightarrow\mathbb{R}$ which we want to act on the eigenvalues of a matrix, then mirror descent iterates are of the form
    \begin{equation}\label{eqn:tr-f-md}
        X^{k+1} = (h')^{-1}[h'(X^k) - t_k\nabla f(X^k)],
    \end{equation}
    where $h':\mathbb{R}\rightarrow\mathbb{R}$ is the first derivative of $h$. This is similar to the case when $\varphi$ is separable, i.e., of the form $\varphi(x)=\sum_{i=1}^n\varphi_i(x_i)$, as discussed in~\cite{bauschke2017descent}. Here, we apply essentially the same principle but to the eigenvalues of the variable $X$.
\end{rem}

We now show how these results can be applied specifically to compute the relative entropy of a quantum resource.
\begin{cor}\label{cor:qre-relative}
    Quantum relative entropy $S\divx{\rho}{\wc}$ with fixed first argument $\rho\in\mathcal{D}(\mathcal{H})$ is $\lambda_{\mathrm{max}}(\rho)$-smooth and $\lambda_{\mathrm{min}}(\rho)$-strongly convex relative to the negative log determinant $-\log(\det(\wc))$ on $\mathcal{D}(\mathcal{H})$.
\end{cor}
\begin{proof}
    This follows from recognizing $f(x)=\log(x)$ is operator concave~\cite[Theorem 2.6]{carlen2010trace}, the identity $\log(\det(\sigma))=\tr[\log(\sigma)]$, and Theorem~\ref{thm:spec-mono}.
\end{proof}

% \begin{prop}\label{cor:qre-relative}
%     Quantum relative entropy $S\divx{\rho}{\wc}$ is $\lambda_{\mathrm{max}}$-smooth and $\lambda_{\mathrm{min}}$-strongly convex relative to the negative log determinant $\varphi(\wc) = -\log(\det(\wc))$ in its second argument, where $\lambda_{\mathrm{max}}$ and $\lambda_{\mathrm{min}}$ is the largest and smallest eigenvalues of $\rho$ respectively.
% \end{prop}
% \begin{proof}
%     For notational convenience, we will prove the desired properties for $g(\sigma) \coloneqq - \tr[\rho\log(\sigma)]$, which apart from a constant is identical to $S\divx{\rho}{\wc}$ and therefore shares the same convexity properties. It follows from joint convexity of quantum relative entropy that $g$ is convex for all $\rho\in\mathbb{H}^n_+$. Therefore, as $\lambda_{\mathrm{max}}\mathbb{I}-\rho\succeq0$, it follows that
%     \begin{equation}
%         -\lambda_{\mathrm{max}}\tr[\log(\sigma)] - g(\sigma) = -\tr[(\lambda_{\mathrm{max}}\mathbb{I} - \rho)\log(\sigma)] ,
%     \end{equation}
%     is convex. Recalling that $\log(\det(\sigma))=\tr[\log(\sigma)]$ for all $\sigma\in\mathbb{H}^n_+$ and the definition for relative smoothness proves the relative smoothness result. Similarly, noting that $\rho-\lambda_{\mathrm{min}}\mathbb{I}\succeq0$ and applying a similar argument recovers the relative strong convexity result.
% \end{proof}

We expect that if $\mathcal{S}(\mathcal{H})$ can be characterized using constraints of the form~\eqref{eqn:constr-constr-min-b} and~\eqref{eqn:constr-constr-min-c}, we can use Algorithm~\ref{alg:pdhg} to compute the relative entropy of a quantum resource. One important example is when we consider $\mathcal{S}(\mathcal{H})$ to be the set of all separable states, in which case~\eqref{eqn:rel-entr-res} recovers the \emph{relative entropy of entanglement}. Although the set of separable states is convex, it is well-known that it is NP-hard to determine if a general quantum state belongs to the set~\cite{gurvits2003classical}. Instead, a relaxation that is commonly used is the positive partial transpose (PPT) criterion~\cite{horodecki1996separability}
\begin{equation}
    \mathsf{PPT}(\mathcal{H}_A\otimes\mathcal{H}_B) = \{ \rho \in \mathcal{D}(\mathcal{H}_A\otimes\mathcal{H}_B) : \rho^{T_B} \succeq 0 \},
\end{equation}
where $\rho^{T_B}$ denotes the partial transpose with respect to system $\mathcal{H}_B$. Membership of $\mathsf{PPT}$ is a necessary condition to be a separable state. For $2\times2$ and $2\times3$ Hilbert systems, this is also a sufficient condition~\cite{horodecki1996separability}. Importantly, the partial transpose is a linear operation, and therefore Algorithm~\ref{alg:pdhg} can be used to solve for the relative entropy of entanglement. To implement this, we use the kernel function $\varphi(\wc)=-\log(\det(\wc))$, the primal domain $\mathcal{C}=\mathcal{D}(\mathcal{H}_A\otimes\mathcal{H}_B)$, and the PPT criterion is dualized. 

Before we introduce the full set of PDHG iterates, we make a few comments on computing the primal step, which can be expressed as
\begin{align}
    \sigma^{k+1} = \argmin_{\sigma \in \mathcal{D}(\mathcal{H}_A\otimes\mathcal{H}_B)} \mleft\{ \inp{\nabla g(\sigma^k) - (\bar{Z}^{k+1})^{T_B}}{\sigma} - \frac{1}{\tau_k}D_\varphi\divx{\sigma}{\sigma^k} \mright\},
\end{align}
where $\varphi(\wc)=-\log(\det(\wc))$. By solving for the KKT conditions and using~\eqref{eqn:tr-f-md} for $h(x)=-\log(x)$, we obtain the following expression
\begin{equation}
    \sigma^{k+1} = [ [\sigma^k]^{-1} + \tau_k( \nabla g(\sigma^k) - (\bar{Z}^{k+1})^{T_B} ) + \nu^{k+1}\mathbb{I} ]^{-1},
\end{equation}
where $\nu^{k+1} \in \mathbb{R}$ is a Lagrange multiplier chosen such that $\sigma^{k+1} \in \mathcal{D}(\mathcal{H}_A\otimes\mathcal{H}_B)$ is satisfied. Note that unlike the case when $\varphi(\wc)=-S(\wc)$, here we need to numerically solve for this $\nu^{k+1}$. To satisfy the unit trace constraint, $\nu^{k+1}$ must satisfy
\begin{equation}\label{eqn:pdhg-ree-d}
    \sum_{i=1}^{nm} \frac{1}{\lambda_i^{k+1} + \nu^{k+1}} = 1,
\end{equation}
where $\lambda_i^{k+1}$ are the eigenvalues of 
\begin{equation}\label{eqn:sigma-bar}
    \bar{\sigma}^{k+1} = [\sigma^{k}]^{-1} + \tau_k(\nabla g(\sigma^k) - (\bar{Z}^{k+1})^{T_B}).
\end{equation}
To satisfy the positivity constraint, $\nu^{k+1}$ must also satisfy $\nu^{k+1} > -\min_i\{\lambda_i^{k+1}\}$. Over this interval, the function on the left-hand side of~\eqref{eqn:pdhg-ree-d} is monotonically decreasing, and takes all values in the open interval $(0, +\infty)$. Therefore, there will always be a unique solution to the equation, which we can efficiently solve for using algorithms such as Newton's method or the bisection method. More details about solving this univariate equation using Newton's method can be found in~\cite[Section 2.5]{jiang2022bregman}. Notably, as~\eqref{eqn:pdhg-ree-d} arises from the optimality conditions of minimizing a self-concordant function, Newton's method with backtracking is guaranteed to converge.

Therefore, we obtain the PDHG iterates $(\sigma^k, Z^k)$ of the form
\begin{subequations}
    \begin{align}
        & \bar{Z}^{k+1} = Z^k + \theta_k (Z^k - Z^{k-1})\\
        & \sigma^{k+1} = [\bar{\sigma}^{k+1} + \nu^{k+1}\mathbb{I} ]^{-1} \label{eqn:pdhg-ree-b}\\
        & Z^{k+1} = (Z^k - \gamma_k (Z^{k+1})^{T_B})_+,
    \end{align}
\end{subequations}
where $\bar{\sigma}^{k+1}$ is given by~\eqref{eqn:sigma-bar}, $\nu^{k+1}\in\mathbb{R}$ is the largest root of~\eqref{eqn:pdhg-ree-d}, and $Z\in\mathcal{B}(\mathcal{H}_A\otimes\mathcal{H}_B)_+$ is the dual variable corresponding to the PPT constraint. 

\begin{rem}
    It is also possible to implement PDHG by splitting the constraints so that we have a primal domain $\mathcal{C}=\mathcal{B}(\mathcal{H}_A\otimes \mathcal{H}_B)_+$, and dualizing both the PPT and unit trace constraint. This would allow us to directly compute the primal step without having to numerically compute for the Lagrange multiplier $\nu^{k+1}$. However, the negative log determinant would no longer be strongly-convex on $\mathcal{C}$, and therefore it is not clear whether there exists a suitable step size which satisfies~\eqref{eqn:pdhg-cond}, which would otherwise allow us to guarantee convergence rates.
\end{rem}

We show experimental results for the computation of the approximate relative entropy of entanglement over PPT states in Table~\ref{tab:ree} and Figure~\ref{fig:ree}. The equation~\eqref{eqn:pdhg-ree-d} is solved using the Newton's method. All states $\rho$ are randomly generated. Results show PDHG solves up to $70$ times faster to medium-accuracy solutions than MOSEK. In Figure~\ref{fig:ree}, we also show convergence results of PDHG for a range of input states $\rho$ with various condition numbers. These states are obtained by linearly interpolating between a randomly generated density matrix with either the maximally mixed state or a rank-one state to acheive a desired condition number. We see that the convergence rate is of PDHG appears to be related to how well conditioned the state $\rho$ is. Notably, we generally see a sublinear convergence rate for instances corresponding to high condition numbers, and linear convergence rates for instances corresponding to small condition numbers. Propositions~\ref{prop:conv-rate} and Corollary~\ref{cor:qre-relative} tell us mirror descent applied to compute the relative entropy of a quantum resource~\eqref{eqn:rel-entr-res} will achieve faster linear convergence the more well-conditioned $\rho$ is, and empirically it seems that PDHG shares similar convergence properties.

\begin{figure*}[]
\centering
\pgfplotstableread[col sep=comma]{figures/data/conv.csv}\data
\begin{tikzpicture}
    \definecolor{clr1}{RGB}{0,0,0}
    \definecolor{clr2}{RGB}{33,33,33}
    \definecolor{clr3}{RGB}{66,66,66}
    \definecolor{clr4}{RGB}{100,100,100}
    \definecolor{clr5}{RGB}{133,133,133}
    
    \begin{semilogyaxis}[
        name=mainplot,
        width = 0.375\textwidth,
        height = 0.375\textwidth,
        xmin = 0.0, xmax = 334,
        ymin = 1e-5, ymax=1,
        xlabel=Iteration $k$,
        ylabel=Optimality gap (nats),
        label style = {font=\footnotesize},
        ticklabel style = {font=\footnotesize},
        title style = {font=\footnotesize},
        legend pos  = north east,
        legend cell align={left},
        legend style={fill=white, fill opacity=0.6, draw opacity=1,text opacity=1, font=\footnotesize},
        no markers,
        every axis plot/.append style={very thick},
        grid=major,
        major grid style={line width=.1pt,draw=gray!20},
        yminorticks=false,
        title = {(a) $n=m=2$}
        ]
        
        \addplot[clr5, dashdotted] table[x=i , y=ree2_10] {\data}; 
        \addplot[clr3, densely dashed] table[x=i , y=ree2_100] {\data};
        \addplot[clr1] table[x=i , y=ree2_1000] {\data};

        \legend{
            {$\bm{\kappa=9\times10^0}$},
            {$\kappa=1\times10^2$},
            {$\kappa=1\times10^3$}
        }

    \end{semilogyaxis}
\end{tikzpicture}
\begin{tikzpicture}
    \definecolor{clr1}{RGB}{0,0,0}
    \definecolor{clr2}{RGB}{33,33,33}
    \definecolor{clr3}{RGB}{66,66,66}
    \definecolor{clr4}{RGB}{100,100,100}
    \definecolor{clr5}{RGB}{133,133,133}

    \begin{semilogyaxis}[
        name=secondplot,
        at={(mainplot.north east)},
        xshift=0.05cm,
        anchor=north west,   
        width = 0.375\textwidth,
        height = 0.375\textwidth,
        xmin = 0.0, xmax = 582,
        ymin = 1e-5, ymax=1,
        xlabel=Iteration $k$,
        label style = {font=\footnotesize},
        ticklabel style = {font=\footnotesize},
        title style = {font=\footnotesize},
        legend pos  = north east,
        legend cell align={left},
        legend style={fill=white, fill opacity=0.6, draw opacity=1,text opacity=1, font=\footnotesize},
        no markers,
        every axis plot/.append style={very thick},
        grid=major,
        major grid style={line width=.1pt,draw=gray!20},
        yminorticks=false,
        yticklabels={,,},
        title = {(b) $n=m=5$}
        ]
        
        \addplot[clr5, dashdotted] table[x=i , y=ree5_10] {\data}; 
        \addplot[clr3, densely dashed] table[x=i , y=ree5_300] {\data};
        \addplot[clr1] table[x=i , y=ree5_1000] {\data};

        \legend{
            {$\kappa=2\times10^1$},
            {$\kappa=3\times10^2$},
            {$\bm{\kappa=2\times10^4}$}
        }
        
    \end{semilogyaxis}
\end{tikzpicture}
\begin{tikzpicture}
    \definecolor{clr1}{RGB}{0,0,0}
    \definecolor{clr2}{RGB}{33,33,33}
    \definecolor{clr3}{RGB}{66,66,66}
    \definecolor{clr4}{RGB}{100,100,100}
    \definecolor{clr5}{RGB}{133,133,133}

    \begin{semilogyaxis}[
        at={(secondplot.north east)},
        xshift=0.28cm,
        anchor=north west,  
        width = 0.375\textwidth,
        height = 0.375\textwidth,
        xmin = 0.0, xmax = 615,
        ymin = 1e-5, ymax=1,
        xlabel=Iteration $k$,
        label style = {font=\footnotesize},
        ticklabel style = {font=\footnotesize},
        title style = {font=\footnotesize},
        legend pos  = north east,
        legend cell align={left},
        legend style={fill=white, fill opacity=0.6, draw opacity=1,text opacity=1, font=\footnotesize},
        no markers,
        every axis plot/.append style={very thick},
        grid=major,
        major grid style={line width=.1pt,draw=gray!20},
        yminorticks=false,
        yticklabels={,,},
        title = {(c) $n=m=9$}
        ]
        
        \addplot[clr5, dashdotted] table[x=i , y=ree9_10] {\data}; 
        \addplot[clr3, densely dashed] table[x=i , y=ree9_1000] {\data};
        \addplot[clr1] table[x=i , y=ree9_100000] {\data};

        \legend{
            {$\kappa=1\times10^1$},
            {$\kappa=7\times10^2$},
            {$\bm{\kappa=7\times10^4}$}
        }
        
    \end{semilogyaxis}
\end{tikzpicture}
% \vspace{-1.5em}
\caption{Convergence of PDHG to compute the relative entropy of entanglement of density matrices with condition numbers $\kappa$ over PPT states for states defined over the tensor product of two Hilbert spaces with dimensions $n$ and $m$. Plots corresponding to bolded legend items correspond to problem instances in Table~\ref{tab:ree}.}
\label{fig:ree}
\end{figure*}
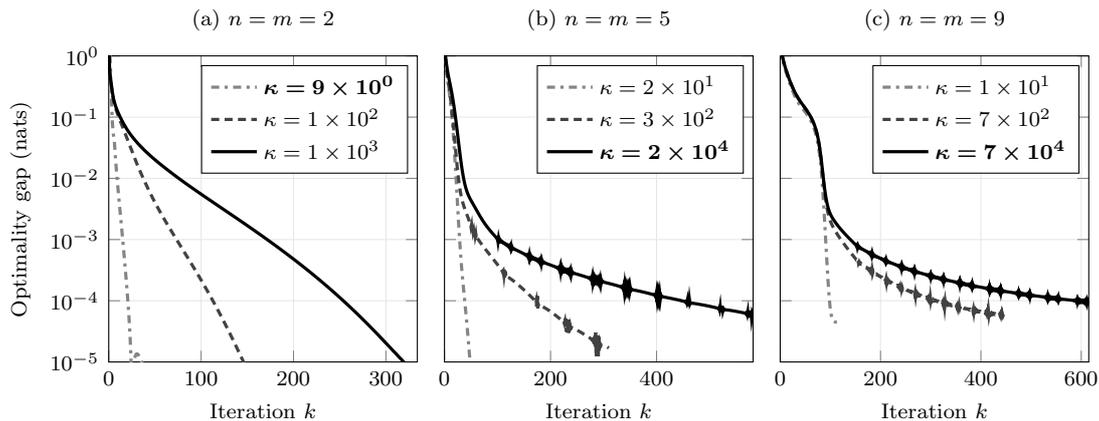

\begin{table}
\caption{Computation of the approximate relative entropy of entanglement over PPT states for states defined over the tensor product of two Hilbert spaces with dimensions $n$ and $m$.}\label{tab:ree}
\vspace{-0.5em}
\footnotesize
\centering
\begin{tabular*}{\columnwidth}{@{\extracolsep{\fill}}ccccccc@{\extracolsep{\fill}}}
\toprule
\multicolumn{1}{c}{} & & \multicolumn{3}{c}{PDHG ($\kappa=1$)} & \multicolumn{2}{c}{MOSEK w/ CVXQUAD} \\ \cmidrule{3-5}\cmidrule{6-7}
$n$ & $m$ & Time (s) & Iter. & Opt. gap (nats) & Time (s) & Opt. gap (nats) \\ \midrule
$2$ & $2$ & $0.020$ & $43$ & $4.0\times10^{-6}$ & $1.400$ & $1.5\times10^{-8}$ \\
$5$ & $5$ & $0.183$ & $582$ & $6.0\times10^{-5}$ & $1.966$ & $1.5\times10^{-8}$ \\
$9$ & $9$ & $2.159$ & $615$ & $9.6\times10^{-5}$ & $91.05$ & $1.5\times10^{-8}$ \\ \bottomrule
\end{tabular*}
\end{table}

\section{Concluding Remarks}

In this work, we have shown that classical and quantum Blahut-Arimoto algorithms can be interpreted as a special application of the mirror descent algorithm, and that existing convergence results are recovered under relative smoothness and strong convexity analysis. This interpretation allows us to extend these algorithms to other applications in information theory, either by using different kernel functions or algorithmic variations such as PDHG which allow us to solve problems with arbitrary linear constraints. 

As the Bregman proximal framework is very general, we believe that it can be applied to many other convex optimization problems in information theory. We also believe that compared to the alternating optimization interpretation traditionally used to derive Blahut-Arimoto algorithms~\cite{arimoto1972algorithm,blahut1972computation,nagaoka1998algorithms,osawa2001numerical,li2019blahut,li2019computing, ramakrishnan2020computing}, the mirror descent interpretation allows for a more straightforward implementation and generalization to other problems, as all it requires is the computation of the objective function's gradient, rather than trying to find a suitable bivariate extension function of which we are not aware of any straightforward way of doing for general problems. 

The main difficulty in implementing Bregman proximal methods on new problems is that identifying problems which are relatively smooth with respect to a suitable kernel function is a non-trivial task in general. Conversely, it can also be difficult to determine a suitable kernel function which a problem is relatively smooth with respect to. Similar issues hold for Blahut-Arimoto-like algorithms following a framework similar to~\cite{ramakrishnan2020computing}. Some works~\cite{li2019convergence,you2022minimizing} have established weak convergence results of mirror descent without requiring Lipschitz gradient nor relative smoothness properties of the objective function. However, if guarantees on convergence rates to the global optimum are desired then these additional assumptions on the objective may be required. Every convex function is clearly $1$-smooth and $1$-strongly convex relative to itself. However, this does not lead to a practical algorithm as solving the mirror descent iterates becomes identical to solving the original problem. Recently, it was shown in~\cite{tsai2023online} how the objective function for the maximum-likelihood quantum state tomography problem is $1$-smooth relative to the negative log-determinant, and in~\cite{tsai2022faster} how an online mirror descent algorithm could be used to solve this problem. Our PDHG framework would allow us to extend this algorithm to account for maximum-likelihood state tomography problems with affine constraints.

As another possible avenue to find problems for which we can implement the Bregman proximal framework, we recall that Theorem~\ref{thm:spec-mono} introduced a generalized method of establishing relative smoothness and strong convexity of a class of functions $X\mapsto \tr[Af(X)]$ for operator convex functions $f$. In Remark~\ref{rem:trace-func-kernel} we also show how mirror descent iterates can be efficiently computed for these functions. We are therefore interested in whether there exists other applications which can utilize these tools. For example, the standard divergences introduced in~\cite{bach2022sum} are defined using operator functions, and can all be analyzed using Theorem~\ref{thm:spec-mono} by noting that $f(x)=x^\alpha$ is operator convex for $\alpha\in[-1,0]\cup[1,2]$ and operator concave for $\alpha\in[0,1]$ (see, e.g.,~\cite{carlen2010trace}).

All of the problems studied in this paper were convex optimization problems. However, there are several important related problems which are posed as non-convex problems. For example, there is interest in computing the cq channel capacity over both the input probability distribution and quantum state alphabet. However, Holevo information is convex in the quantum state alphabet, making the objective function non-convex. Other similar problems include the classical~\cite{tishby2000information} and quantum~\cite{datta2019convexity} information bottleneck functions, which involve non-convex inequality constraints. One common way to find local optima of general non-convex problems is through a convex-concave decomposition, of which there exist variations which utilize proximal gradient iterations~\cite{souza2016global,wen2018proximal}. Therefore, it will be interesting to see if existing algorithms which solve for these quantities~\cite{nagaoka1998algorithms,tishby2000information,salek2018quantum,hayashi2023efficient} share the same interpretation or can be improved by a Bregman proximal convex-concave decomposition or similar method.

{\appendix

\section{Backtracking Primal-Dual Hybrid Gradient} \label{appdx:bcktrack-pdhg}

An implementation issue with PDHG is that it can be difficult to determine suitable step sizes $\tau$ and $\gamma$ which satisfy~\eqref{eqn:pdhg-cond}. If the kernel function $\varphi$ is strongly convex, then the simplified condition~\eqref{eqn:pdhg-cond-alt} provides an easier way to choose these step sizes. However, obtaining a tight bound on the relative smoothness parameter $L$ or computing $\norm{A}$ may be non-trivial tasks. Even when these constants can be easily computed, step sizes obtained from~\eqref{eqn:pdhg-cond-alt} may be too conservative. 

To resolve this issue, we introduce a backtracking method summarized in Algrotihm~\ref{alg:pdhg-backtrack} which adaptively chooses step sizes to satisfy conditions required for convergence. Notably, the condition does not require us to know $L$ or $\norm{\mathcal{A}}$, and is therefore easily computable. This algorithm is based on the backtracking method introduced in~\cite{jiang2022bregman} which accounts for unknown $\norm{\mathcal{A}}$. Our algorithm also accounts for unknown $L$ by using a similar approach as~\cite[Section 5]{malitsky2018first}. Note that the backtracking exit criterion~\eqref{eqn:pdhg-exit} can be interpreted as a combination of condition~\eqref{eqn:pdhg-cond} and relative smoothness as characterized by Proposition~\ref{prop:relative-alt}(a-iii). 

\begin{algorithm*}
    \caption{Backtracking primal-dual hybrid gradient}
    \begin{algorithmic}
        \Input{Objective function $f$, Legendre reference function $\varphi$, primal domain $\mathcal{C}\subseteq\domain\varphi$,}
        \Indent{linear constraint parameters $\mathcal{A}, b$, backtracking parameters $\alpha<1$, $\bar{\theta}\geq1$.}
        \State \textbf{Initialize:} Initial points $x^{-1}, x^0\in\mathcal{C}$, $z^0,z^{-1}\in\mathcal{Z}$, step sizes $\tau_{-1}, \gamma_{-1}>0$.
        \For{$k=0,1,\ldots$}{}
            \State \textbf{Step 0}: Initialize $\theta_k = \Bar{\theta}$.
            \State \textbf{Step 1}: Compute step sizes $\tau_k = \theta_k \tau_{k-1}$ and $\gamma_k = \theta_k \gamma_{k-1}$.
            \State \textbf{Step 2}: Perform primal-dual hybrid gradient steps
            \begin{subequations}
                \begin{align}
                    &\Bar{z}^{k+1}= z^k + \theta_k(z^k - z^{k-1})\\
                    &\displaystyle x^{k+1} = \argmin_{x \in \mathcal{C}} \mleft\{ \inp{\nabla f(x) + \mathcal{A}^\dag (\Bar{z}^{k+1})}{x} + \frac{1}{\tau_k} D_\varphi\divx{x}{x^k} \mright\}\\
                    &\displaystyle z^{k+1} = \argmin_{z \in \mathcal{Z}} \biggl\{ -\inp{z}{\mathcal{A}(x^{k+1}) - b} + \frac{1}{2\gamma_k} \norm{z - z^k}_2^2 \biggr\}
                \end{align}
            \end{subequations}
            \State \textbf{Step 3}: Check if backtracking exit criterion is satisfied
            \begin{gather}
                \begin{aligned}
                    \MoveEqLeft[8] f(x^{k+1}) - f(x^k) - \inp{\nabla f(x^k)}{x^{k+1} - x^k} \\
                    &\leq \frac{1}{\tau_k} D_\varphi\divx{x^{k+1}}{x^k} + \frac{1}{2\gamma_k}\norm{z^{k+1} - \Bar{z}^{k+1}}_2^2 - \inp{z^{k+1} - \Bar{z}^{k+1}}{\mathcal{A}(x^{k+1} - x^k)}. \label{eqn:pdhg-exit}
                \end{aligned}                
            \end{gather}
            \State \hspace{3em} If not, then update $\theta_k \coloneqq \alpha \theta_k$ and go to Step 1.
        \EndFor
        \Output{Approximate primal-dual solution $(x^{k+1}, z^{k+1})$.}
    \end{algorithmic}
    \label{alg:pdhg-backtrack}
\end{algorithm*}

We now present the following convergence result for the backtracking algorithm.

\begin{prop}\label{prop:pdhg-terminate}
    Consider Algorithm~\ref{alg:pdhg-backtrack} to solve the convex optimization problem~\eqref{eqn:constr-constr-min}. If $f$ is $L$-smooth relative to $\varphi$ and $\varphi$ is strongly convex with respect to some norm, then the step sizes $(\tau_k, \gamma_k)$ are bounded below by
    \begin{align*}
        &\tau_k \geq \tau_{\mathrm{min}} \coloneqq \min \biggl\{ \tau_{-1}, \alpha\biggl( \sqrt{ \frac{L^2\kappa^2}{4\norm{\mathcal{A}}^4} + \frac{\kappa}{\norm{\mathcal{A}}^2} } - \frac{L\kappa}{2\norm{\mathcal{A}}^2} \biggr) \biggr\}\\
        &\gamma_k \geq \gamma_{\mathrm{min}} \coloneqq \tau_{\mathrm{min}} / \kappa,
    \end{align*}
    where $\kappa \coloneqq \tau_{-1} / \gamma_{-1}$. It also follows that the backtracking procedure will always terminate.
\end{prop}
\begin{proof}
    As the backtracking exit condition~\eqref{eqn:pdhg-exit} is the sum of~\eqref{eqn:pdhg-cond} and the relative smoothness condition in Proposition~\ref{prop:relative-alt}(a-iii), a sufficient condition for the exit condition to hold is for these two conditions to hold independently. As $f$ is $L$-smooth relative to $\varphi$, a sufficient condition is therefore for~\eqref{eqn:pdhg-cond-alt} to hold. The remaining proof follows from a similar argument as~\cite[Section 3.3.1]{jiang2022bregman}.
\end{proof}

\begin{prop}\label{prop:appdx-pdhg-conv}
    Consider Algorithm~\ref{alg:pdhg-backtrack} to solve the convex optimization problem~\eqref{eqn:constr-constr-min}. Let $f^*$ represent the optimal value of this problem and $(x^*, z^*)$ be any corresponding optimal primal-dual solution. If $f$ is $L$-smooth relative to $\varphi$, then the iterates $(x^k, z^k)$ satisfy
    \begin{align}
        \mathcal{L}(x_{\mathrm{avg}}^k, z) - \mathcal{L}(x, z_{\mathrm{avg}}^k) \leq \frac{1}{k\tau_{\mathrm{min}}} \mleft( D_\varphi\divx{x}{x^0} + \frac{1}{2\kappa}\norm{z - z^0}^2_2 \mright),
    \end{align}
    for all $x\in\mathcal{C}$, $z\in\mathcal{Z}$, and $k\in\mathbb{N}$, where $\tau_{\mathrm{min}}=\min_i\{\tau_i\}$,
    \begin{equation*}
        x_{\mathrm{avg}}^k = \frac{1}{\sum_{i=1}^{k} \tau_{i-1}} \sum_{i=1}^{k} \tau_{i-1} x^i, 
    \end{equation*}
    and
    \begin{equation*}
        z_{\mathrm{avg}}^k = \frac{1}{\sum_{i=1}^{k} \tau_{i-1}} \sum_{i=1}^{k} \tau_{i-1} \Bar{z}^i.
    \end{equation*}
\end{prop}
\begin{proof}
    Using the Bregman proximal inequality~\cite[Lemma 3.2]{chen1993convergence}, we first establish that
    \begin{align}
        \inp{\nabla f(x^k) + \mathcal{A}^\dag(\Bar{z}^{k+1})}{x^{k+1} - x} \leq \frac{1}{\tau_k} (D_\varphi\divx{x}{x^k} - D_\varphi\divx{x^{k+1}}{x^k} - D_\varphi\divx{x}{x^{k+1}}), \label{eqn:pdhd-p-bpi}
    \end{align}
    for all $x\in\mathcal{C}$, and
    \begin{align}
        \inp{z - z^{k+1}}{\mathcal{A}(x^{k+1}) - b} \leq \frac{1}{\gamma_k} \inp{z^{k+1} - z^k}{z - z^{k+1}} ,\label{eqn:pdhg-proof-b}
    \end{align}
    for all $z\in\mathcal{Z}$. Using this second result, we can use a similar argument as~\cite{jiang2022bregman} to show that
    \begin{align}
        \inp{z^{k+1} - \Bar{z}^{k+1}}{\mathcal{A}(x^k) - b} \leq \frac{1}{2\gamma_k} (\norm{z^{k+1} - z^k}^2_2 - \norm{z^{k+1} - \Bar{z}^{k+1}}^2_2 - \norm{\Bar{z}^{k+1} - z^{k}}^2_2). \label{eqn:pdhg-proof-c}
    \end{align}
    Now, subtracting $f(x)$ for an arbitrary $x\in\mathcal{C}$ from both sides of the exit condition~\eqref{eqn:pdhg-exit}, using convexity of $f$, and substituting in~\eqref{eqn:pdhd-p-bpi} gives
    \begin{align}
        f(x^{k+1}) - f(x) + \inp{\Bar{z}^{k+1}}{\mathcal{A}(x^{k+1} - x)} &\leq \frac{1}{\tau_k} (D_\varphi\divx{x}{x^k} - D_\varphi\divx{x}{x^{k+1}}) + \frac{1}{2\gamma_k}\norm{z^{k+1} - \Bar{z}^{k+1}}_2^2 \nonumber\\ 
        & \hphantom{{}\leq{}} - \inp{z^{k+1} - \Bar{z}^{k+1}}{\mathcal{A}(x^{k+1} - x^k)}. \label{eqn:pdhg-proof-d}
    \end{align}
    By combining~\eqref{eqn:pdhg-proof-b}--\eqref{eqn:pdhg-proof-d} and using some algebraic manipulation, we can obtain
    \begin{align}
        &\mathcal{L}(x^{k+1}, z) - \mathcal{L}(x, \Bar{z}^{k+1}) \leq \frac{1}{\tau_k} (D_\varphi\divx{x}{x^k} - D_\varphi\divx{x}{x^{k+1}}) + \frac{1}{2\gamma_k} (\norm{z - z^k}^2_2 - \norm{z - z^{k+1}}^2_2). \label{eqn:pdhg-proof-z}
    \end{align}
    The rest of the proof follows from~\cite[Section 3.3.3]{jiang2022bregman}, which we include here to make the proof self-contained. Using convexity of $\mathcal{L}(\wc, z)$ and concavity of $\mathcal{L}(x, \wc)$, we get
    \begin{align*}
        \mathcal{L}(x^k_\mathrm{avg}, z) - \mathcal{L}(x, z^k_\mathrm{avg}) &\leq \frac{1}{\sum_{i=1}^{k} \tau_{i-1}} \sum_{i=1}^{k} \tau_{i-1}(\mathcal{L}(x^i, z) - \mathcal{L}(x, \Bar{z}^i))\\
        & \leq \frac{1}{k\tau_{\mathrm{min}}} \sum_{i=1}^{k} \biggl(D_\varphi\divx{x}{x^{i-1}} - D_\varphi\divx{x}{x^i}  + \frac{1}{2\kappa} (\norm{z - z^{i-1}}^2_2 - \norm{z - z^i}^2_2)\biggr).
    \end{align*}
    Taking a telescoping sum and using non-negativity of Bregman divergences produces the desired result.
\end{proof}

\section{Matrix-valued Gradients}\label{appdx:grad}

To use mirror descent to solve problems in quantum information theory, we require expressions for the gradients of functions with matrix-valued inputs. We introduce the tools from matrix analysis which allow us to compute these. 
\begin{defn}[Operator function]
Let $f: \mathbb{R}\rightarrow\bar{\mathbb{R}}$ be a continuous extended-real-valued function, and $X\in\mathbb{H}^{n}$ be a Hermitian matrix with spectral decomposition $X=\sum_i\lambda_iv_iv_i^\dag$. The function $f$ can be extended to matrices $X$ as follows
\begin{equation*}
    f(X) \coloneqq \sum_{i=1}^nf(\lambda_i)v_iv_i^\dag.
\end{equation*}
\end{defn}

\begin{rem}
    The matrix logarithms used to define von Neumann entropy and quantum relative entropy use precisely this definition of primary matrix functions. In particular, von Neumann entropy can be defined as $S(X)=\tr[f(X)]$ where $f(x)=-x\log(x)$. 
\end{rem}

\begin{lem}[{\cite[Theorem 6.6.30]{horn1994topics} and \cite[Theorem 3.23]{hiai2014introduction}}]\label{lem:dir-deriv}
    Let $f: \mathbb{R}\rightarrow\bar{\mathbb{R}}$ be a continuously differentiable scalar-valued function with derivative $f'$ defined on an open real interval $(a, b)$. Consider $X, V \in \mathbb{H}^n$ where $X$ has diagonalization $X=U\Lambda U^\dag$ with $\Lambda=\diag(\lambda_1,\ldots,\lambda_n)$, and $\lambda_i\in(a, b)$ for all $i=1,\ldots,n$.
    \begin{enumerate}[label=(\alph*)]
        \item The directional derivative of $f(X)$ along $V$ is
            \begin{equation}
                \mathsf{D}f(X)[V] = U [ f^{[1]}(\Lambda) \odot (U^\dag V U) ] U^\dag,
            \end{equation}
            where $\odot$ represents the Hadamard or element-wise product and $f^{[1]}(\Lambda)$ is the first-divided difference matrix whose $(i, j)$-th entry is given by $f^{[1]}(\lambda_{i}, \lambda_{j})$ where
            \begin{subequations}\label{eqn:fdd}
                \begin{align}
                    f^{[1]}(\lambda, \mu) &= \frac{f(\lambda) - f(\mu)}{\lambda - \mu}, \quad \textrm{if }\lambda\neq\mu, \\
                    f^{[1]}(\lambda, \lambda) &= f'(\lambda).
                \end{align}
            \end{subequations}
        \item The directional derivative of the trace functional $g(X)= \tr [f(X)]$ along $V$ is
            \begin{align}
                \mathsf{D}g(X)[V] = \tr[f'(X) V].
            \end{align}
    \end{enumerate}
\end{lem}

The following corollaries are straightforward consequences of the gradient being defined by $\mathsf{D}g(X)[V] = \inp{\nabla g(X)}{V}$ for all $V\in\mathbb{H}^n$.
\begin{cor}\label{cor:grad-tr-f}
For $X\in\interior\domain f$, the gradient of the trace functional $g(X)= \tr [f(X)]$ is
\begin{equation*}
    \nabla g(X) = f'(X).
\end{equation*}
\end{cor}
\begin{cor}\label{cor:grad-tr-f-lin}
Let $\mathcal{A}:\mathbb{H}^n\rightarrow\mathbb{H}^m$ be an affine operator such that $\mathcal{A}(X) = \mathcal{L}(X) + C$ where $\mathcal{L}:\mathbb{H}^n\rightarrow\mathbb{H}^m$ is a linear operator and $C\in\mathbb{H}^m$ is a constant matrix. For $X\in\interior\domain h$, the gradient of $h(X) = \tr[f(\mathcal{A}(X))]$ is
\begin{equation*}
    \nabla h(X) = \mathcal{L}^\dag(f'(\mathcal{A}(X))).
\end{equation*}
\end{cor}

\section{Proof for Proposition~\ref{prop:log-strong}}\label{sec:log-strong}

We will only show the proof for the negative log determinant function $\varphi(X)=\tr[f(X)]$ for $f(x)=-\log(x)$, as it generalizes the Burg entropy. Using the second derivative of the log determinant~\cite[Appendix A.4.3]{boyd2004convex}, $\varphi$ is $\mu$-strongly convex with respect to the Frobenius norm only if
\begin{equation*}
    \mathsf{D}^2 \varphi(X)[V, V] = \tr[X^{-1}VX^{-1}V] \geq \mu,
\end{equation*}
for all $X\in\mathcal{D}(\mathcal{H})_{++}$ and $V \in \{ V\in\mathcal{B}(\mathcal{H}) : \norm{V}_2 = 1 \}$. Note that for all matrices $X\in\mathcal{D}(\mathcal{H})_{++}$, $0\prec X\preceq\mathbb{I}$, and therefore $X^{-1}\succeq\mathbb{I}$ and $VX^{-1}V\succeq0$ for all Hermitian $V$. It then follows from the fact that the trace of the product of two positive semidefinite matrices is always non-negative that 
\begin{equation*}
    \mathsf{D}^2 \varphi(X)[V, V] \geq \tr[VX^{-1}V] \geq \tr[VV] = \norm{V}_2^2 = 1.
\end{equation*}
Therefore, $\mu=1$ satisfies the inequality, which concludes the proof.
}

\bibliographystyle{IEEEtran}
\bibliography{bibliofile}

\end{document}